\documentclass[]{article}

\usepackage[utf8]{inputenc}
\usepackage{microtype}
\usepackage{libertine}
\usepackage{libertinust1math}
\usepackage[T1]{fontenc}
\usepackage{amssymb}
\usepackage{amsfonts}
\usepackage{dsfont}
\usepackage{amsmath}
\usepackage{amsthm}
\usepackage{tikz}
\usetikzlibrary{automata,positioning}

\usepackage[margin=1cm]{caption}

\usepackage[english]{babel}

\RequirePackage[l2tabu, orthodox]{nag}
\usepackage[all,warning]{onlyamsmath}

\usepackage{hyperref}
\usepackage{cleveref}
\usepackage{color}
\usepackage{mathtools}
\usepackage{xspace}

\newtheorem{theorem}{Theorem}[section]

\newtheorem{proposition}[theorem]{Proposition}
\newtheorem{lemma}[theorem]{Lemma}
\newtheorem{definition}[theorem]{Definition}
\newtheorem{remark}[theorem]{Remark}
\newtheorem{observation}[theorem]{Observation}

\DeclareMathOperator{\init}{init}
\DeclareMathOperator{\out}{out}
\DeclareMathOperator{\up}{up}
\newcommand{\mem}{\mathcal{M}}

\DeclareMathOperator{\last}{last}
\newcommand{\distrib}[1]{\Delta(#1)}

\newcommand{\states}{\mathbf{S}}
\newcommand{\actions}{\mathbf{A}}

\DeclareMathOperator{\val}{val}
\DeclareMathOperator{\maxmin}{maxmin}
\DeclareMathOperator{\minmax}{minmax}

\newcommand{\ee}[3]{\mathbb E_{#1}^{#2}\left[#3\right]}
\newcommand{\eed}[1]{\ee s {\sigma,\tau} {#1}}

\newcommand{\pp}[3]{\mathbb{P}_{#1}^{#2}\left({#3}\right)}
\newcommand{\ppd}[1]{\pp s {\sigma,\tau} {#1}}
\newcommand{\ppdd}{\mathbb P_s^{\sigma,\tau}}

\newcommand{\arena}{\A}

\newcommand{\reset}[1]{\hat{#1}}

\newcommand{\colors}{\mathbf C}

\newcommand{\psmall}[2]{\psmalln\left(#1,#2\right)}
\newcommand{\psmalln}{p}

\newcommand{\pay}{f}
\newcommand{\mdisc}{\pay_{\text{disc}}}
\newcommand{\parite}{\pay_\text{par}}
\newcommand{\paylimsup}{\pay_\text{lsup}}
\newcommand{\payliminf}{\pay_\text{linf}}
\newcommand{\payposavg}{\pay_\text{posavg}}

\newcommand{\mean}{\pay_\text{mean}}

\newcommand{\zopt}{\sharp}
\newcommand{\sopt}{\ensuremath{\sigma^\zopt}}
\newcommand{\topt}{\ensuremath{\tau^\zopt}}

\newcommand{\game}{\ensuremath{\mathbf{G}}}

\newcommand{\pref}{\ensuremath{\preceq}}
\newcommand{\spref}{\ensuremath{\prec}}

\newcommand{\nat}{\mathbb N}
\newcommand{\rel}{\mathbb R}

\newcommand{\A}{\ensuremath{\mathcal{A}}} %automaton

\newcommand{\ie}{{\em i.e.}\xspace}
\newcommand{\eg}{{\em e.g.}\xspace}

\newcommand{\set}[1]{\{#1\}}
\newcommand{\st}{\ :\ }

\newcommand\defeq{\mathrel{\overset{\makebox[0pt]{\mbox{\normalfont\tiny\sffamily def}}}{=}}}

\DeclareMathOperator{\Stay}{Stay}

\newcommand{\stayone}{\Stay_\omega(\game_1)}
\newcommand{\staynone}{\Stay_{\geq n}(\game_1)}
\newcommand{\staytwo}{\Stay_\omega(\game_2)}
\newcommand{\stayntwo}{\Stay_{\geq n}(\game_2)}
\newcommand{\stayntwog}[1]{\Stay_{\geq #1}(\game_2)}
\DeclareMathOperator{\switch}{Switch}

\usepackage{thmtools}
\usepackage{authblk}
\begin{document}

\title{Submixing and Shift-Invariant Stochastic Games}

\author[1]{Hugo Gimbert}
\affil[1]{CNRS, LaBRI, Universit\'e de Bordeaux, France}
%\email{hugo.gimbert@cnrs.fr}
%\thanks{This work was supported by the ANR projet "Stoch-MC" and the LaBEX "CPU".}

\author[2]{Edon Kelmendi}
\affil[2]{Max Planck Institute for Software Systems, Germany}
%\email{edon.kelmendi@cs.ox.ac.uk}

\maketitle

\begin{abstract}
  We study optimal strategies in two-player stochastic games that are played on a finite graph, equipped with a general payoff function. The existence of optimal strategies that do not make use of memory and randomisation is a desirable property that vastly simplifies the algorithmic analysis of such games. Our main theorem gives a sufficient condition for the maximizer to possess such a simple optimal strategy. The condition is imposed on the payoff function, saying the payoff does not depend on any finite prefix (shift-invariant) and combining two trajectories does not give higher payoff than the payoff of the parts (submixing). The core technical property that enables the proof of the main theorem is that of the existence of $\epsilon$-subgame-perfect strategies when the payoff function is shift-invariant. Furthermore, the same techniques can be used to prove a finite-memory transfer-type theorem: namely that for shift-invariant and submixing payoff functions, the existence of optimal finite-memory strategies in one-player games for the minimizer implies the existence of the same in two-player games. We show that numerous classical payoff functions are submixing and shift-invariant. 
\end{abstract}
\maketitle
\newpage
\tableofcontents
\newpage
\section{Introduction}
\label{sec:intro}

The games that we study are played between two players on a finite graph. Every vertex of the graph belongs to one of the players, the one that decides which edge should be taken next. The result of such a play is an infinite path in the graph. The objective of the game is given using a payoff function, which maps infinite paths to real numbers. The maximizer or Player 1, wants to maximize the payoff, while his adversary (the minimizer) wants the opposite.

The study of such games has been an active area of research for a few decades, in a variety communities; especially in that of theoretical computer science and economics. They are used to model simplified adversarial (zero-sum) situations. In computer science they are used in verifying properties of systems, but also as a very beneficial theoretical tool in logic and automata theory.

In this paper we consider \emph{stochastic} games, a more general model where in every step, after an action is chosen, there is a probability distribution on the set of vertices according to which the next vertex is chosen. In this scenario, Player 1 wants to maximize the \emph{expected~payoff}, and his adversary to minimize it. 

Well-known examples of games played on graphs are the discounted games, mean-payoff games, games equipped with the limsup payoff function and parity games. These four classes of games share a common property: both players have very simple optimal strategies, namely optimal strategies that are both deterministic and stationary. These are strategies that guarantee maximal expected payoff and choose actions deterministically (without randomisation) and this deterministic choice depends only on the current vertex (it does not use memory). When games admit such strategies for the maximizer they are called \emph{half-positional}, when they admit such strategies for both players they are called \emph{positional}. This property is highly desirable and it is often the starting point for further algorithmic analysis.

The broad purpose of the present paper is to study what is the common quality of games that makes it possible for them to admit deterministic and stationary optimal strategies.

\paragraph*{Context.}
There have been numerous papers about the existence of deterministic and stationary optimal strategies in games with different payoff functions.  Shapley proved that stochastic games with discounted payoff function are positional using an operator approach~\cite{shapley}.  Derman showed the positionality of one-player games with expected mean-payoff reward, using an Abelian theorem and a reduction to discounted games~\cite{derman}. Gilette extended Derman's result to two-player games~\cite{gilette} but his proof was found to be wrong and corrected by Ligget and Lippman~\cite{liggettlippman}.  The positionality of one-player parity games was addressed in~\cite{CourYan:1990} and later on extended to two-player games in~\cite{ChJurHen:2003b, zielonka:2004}.  Counter games were extensively studied in~\cite{onecountergames} and several examples of positional counter games are given.  There are also several examples of one-player and two-player positional games in~\cite{1jpos, Zielonka10}.  A whole zoology of half-positional games is presented in~\cite{thesekop} and another example is given by mean-payoff co-B{\"u}chi games~\cite{meanpayoffparity}.  The proofs of these various results are quite heterogeneous, making it difficult to find a common property that explains why they are positional or half-positional.

Some effort has been made to better understand conditions that make games (half) positional, which has made apparent that payoff functions that are shift-invariant and submixing play a crucial role. Our contributions lie in this direction.

\paragraph*{Contributions.}
The results of the present paper can be summarised as follows. 

First, the main theorem says that a sufficient condition for the game to be half-positional is for the payoff function to be shift-invariant and submixing. We give an informal explanation of this condition. Payoff functions $f$ map infinite paths of the graph 
\begin{align*}
  s_0s_1s_2s_3\cdots
\end{align*}
to real numbers. A payoff function is \emph{shift-invariant} if it does not depend on finite prefixes, in other words
\begin{align*}
  f(p\ s_0s_1s_2s_3\cdots) = f(s_0s_1s_2s_3\cdots),
\end{align*}
for any finite prefix $p$, \ie we can shift the trajectory to the left without changing the payoff. A payoff function is \emph{submixing} on the other hand, if for any two infinite paths
\begin{align*}
  &{\color{blue}s_0s_1s_2s_3\cdots}\\
  &{\color{red}t_0t_1t_2t_3\cdots}
\end{align*}
shuffling (or combining) them such as
\begin{align*}
  {\color{blue}s_0s_1s_2}&&{\hspace{-1cm}\color{blue}s_3s_4}&&{\color{blue}s_5s_6s_7s_8} \cdots\\
   &\hspace{.5cm}{\color{red}t_0t_1}&&\hspace{1cm}{\color{red}t_2t_3t_4t_5t_6}&&{\color{red}t_7t_8} \cdots
\end{align*}
does not give better payoff, that is:
\begin{align*}
  f({\color{blue}s_0s_1s_2}{\color{red}t_0t_1}{\color{blue}s_3s_4}{\color{red}t_2t_3t_4t_5t_6}{\color{blue}s_5s_6s_7s_8}{\color{red}t_7t_8}\cdots) \le \max\set{f({\color{blue}s_0s_1s_2\cdots}),f({\color{red}t_0t_1t_2\cdots})}.
\end{align*}

\begin{restatable}{theorem}{main}
  \label{thm:main}
  Games equipped with a payoff function that is shift-invariant and submixing are half-positional.
\end{restatable}
As mentioned above, half-positional games are those where the \emph{maximizer} has a simple kind of strategy that is optimal. There is nothing special about this player, if instead of the submixing condition, we define an ``inverse'' submixing condition, namely one that requires that the combined payoff is larger than the minimum of the parts, we would have an analogous theorem that proves the existence of simple optimal strategies for the \emph{minimizer}. Furthermore there are payoff functions for which both versions of the submixing condition hold, and for these games the theorem proves positionality. The conditions in the statement of the theorem are not necessary; we will provide examples and discuss this fact. The proof of \Cref{thm:main} is by induction on number of edges, it uses L\'evy's 0-1 law, as well as the following crucial property of the games under consideration. Namely that games equipped with a payoff function that is both bounded and Borel-measurable admit $\epsilon$-subgame-perfect strategies, for every $\epsilon>0$. A proof of this fact can be found in~\cite{correlated}. 

The second contribution says that having a shift-invariant payoff function is sufficient for the existence of $\epsilon$-subgame-perfect strategies.

\begin{restatable}{theorem}{subgameperfect}
  \label{thm:subgameperfect}
  Games equipped with a payoff function that is shift-invariant, for every $\epsilon>0$, admit $\epsilon$-subgame-perfect strategies.
\end{restatable}
The proof of this theorem uses martingale theory, and takes a large part of the paper, however it is independent of the rest.

A third contribution comes as a corollary of the techniques developed for the main theorem. It is a transfer-type theorem that lifts the existence of optimal finite-memory strategies in one-player games (also known as Markov decision processes) to the same for two-player games. 

\begin{restatable}{theorem}{memory}
  \label{thm:memory}
    \label{thm:mainfinitemem}
Let $f$ be a payoff function that is both shift-invariant and submixing.

Assume that in all games equipped with~$f$ and fully controlled by the minimizer,
for every $\epsilon>0$,
the minimizer has an $\epsilon$-optimal strategy with finite memory.
Then in every (two-player) game, for every $\epsilon>0$, the minimizer has an $\epsilon$-subgame-perfect strategy that has finite memory.

 The statement also holds for $\epsilon=0$, that is: if the minimizer has an optimal strategy with finite memory in every game that he fully controls, then in every (two-player) game as well he has a subgame-perfect strategy with finite memory. 
\end{restatable}

Furthermore this theorem is proved by effectively constructing the $\epsilon$-subgame-perfect strategies in the two-player games. Those are obtained by combining and simplifying $\epsilon$-optimal strategies in one-player games.

A more general result about the transfer of simple class of strategies
for the minimizer from one-player to two-player games is also formulated in \Cref{thm:generaltransfer}.

\paragraph*{Related work.}
For one-player games it was proved by the first author that every one-player game equipped with a payoff function that is both shift-invariant and submixing is positional~\cite{1jpos}. This result was successfully used in~\cite{onecountergames} to prove positionality of counter games. A weaker form of this condition was presented in~\cite{gimbert:zielonka:2004a} to prove positionality of deterministic games (\ie games where transition probabilities are equal to $0$ or $1$, not stochastic). Kopczynski proved that two-player deterministic games equipped with a shift-invariant and submixing payoff function that takes only two values is half-positional~\cite{kopcsl}.

A result of Zielonka~\cite{Zielonka10} provides a necessary and sufficient condition for the positionality of one-player games. The condition is expressed in terms of the existence of particular optimal strategies in multi-armed bandit games.  When trying to prove the positionality for a particular payoff function, the condition in~\cite{Zielonka10} is harder to check than the submixing property which is purely syntactic.

% \Cref{thm:subgameperfect} was proved independently and in parallel by Mashiah-Yaakovi in~\cite[Proposition 11]{correlated} for concurrent games (which subsume the turn-based games that we are considering). However our proof can be seen as a sharpening of the result of Mashiah-Yaakovi for turn-based games, since our construction makes more transparent the fact that building $\epsilon$-subgame-perfect strategies from $\epsilon$-strategies conserves important properties such as having finite memory and not using randomisation. 

Some results on finite-memory determinacy have been obtained in ~\cite{DBLP:conf/concur/Bouyer0ORV20}, with different requirements: the size of the memory should be independent from the arena, whereas in this paper we do not make such an assumption.

The pre-print version of this present paper~\cite{preprintversion} has already been used in a number of works, mostly pertaining the algorithmic game theory community. We mention the papers that we are aware of. In \cite{chatterjee2016}, Chatterjee and Doyen study payoff functions that are a conjunction of mean-payoff objectives, and prove that they are in co-NP for finite-memory strategies. They use \Cref{thm:main}; and for \Cref{thm:subgameperfect} they observe that in the special case of finite-memory strategies there is a simple combinatorial proof, which bypasses the use of martingale theory. In~\cite{basset2018} the authors consider arbitrary boolean combination of expected mean-payoff objectives and the main theorem of the present paper appears as Theorem~1, and is the starting point of their further algorithmic analysis. Games played on finite graphs where the information flow is perturbed by non-deterministic signalling delays are considered in~\cite{berwanger2015}, where submixing and shift-invariant payoff functions play a central r\^ole. Our results and proof techniques were also used by Mayr, Schewe, Totzke and Wojtczak to establish a finite-memory transfer theorem analogous to the second part of Theorem~\ref{thm:memory} and to prove that games with energy-parity objectives and almost-sure semantics lie in NP $\cap$ co-NP~\cite{mayr2021simple}.

\paragraph*{Organisation of the paper.}

We fix the notation and give the relevant definitions in \Cref{sec:prelims}, where one can also find an overview of the proof. We give examples of shift-invariant and submixing payoff functions in \Cref{sec:applications}, as well as show how the \Cref{thm:main} can be used to recover numerous classical determinacy results. In~\Cref{sec:subgameperfect}, we define reset strategies as a method of obtaining $\epsilon$-subgame-perfect strategies, which exist due to~\Cref{thm:subgameperfect}. The proof of the main theorem, \Cref{thm:main}, is given in \Cref{sec:main theorem}, and that of the transfer theorem for finite-memory strategies, \Cref{thm:memory}, in \Cref{sec:memory}.

%%% Local Variables:
%%% mode: latex
%%% TeX-master: "main"
%%% End:

\section{Preliminaries}
\label{sec:prelims}

The purpose of this section is to introduce the basic notions that we need about stochastic games with perfect information, that is the definitions of:  games, payoff functions, strategies and values. 
\paragraph{Games}

A game is specified by the \emph{arena} and the \emph{payoff function}.  While the arena determines \emph{how} the game is played, the payoff function specifies the \emph{objectives} that the players want to reach.

We use the following notations throughout the paper. Let $\states$ be a finite set.  The set of finite (respectively infinite) sequences on $\states$ is denoted $\states^*$ (respectively $\states^\omega$). A \emph{probability distribution} on $\states$ is a function $\delta: \states \to [0,1]$ such that $\sum_{s\in\states} \delta(s) =1$.  The set of probability distributions on $\states$, we denote by $\distrib\states$.

\begin{definition}[Arena]
  A stochastic arena with perfect information is a tuple:
  \begin{align*}
    \left(\states,\  \states_1,\  \states_2,\ \actions,\ \left(\actions(s)\right)_{s\in\states},\ \psmalln \right)
  \end{align*}
  where
\begin{itemize}
  \item $\states$ is a finite set of states (that is nodes of the graph) partitioned in two sets $(\states_1,\states_2)$,
  \item $\actions$ is a finite set of actions,
  \item for each state $s\in\states$, a non-empty set $\actions(s)\subseteq\actions$ of actions available in $s$,
  \item and transition probabilities $p:\states\times\actions\to\distrib{\states}$.
\end{itemize}
\end{definition}
An arena is \emph{fully controlled by the minimizer} if $\actions(s)$ is a singleton for every $s\in\states_1$.

An \emph{infinite play} in an arena $\arena$ is an infinite sequence $p=s_0a_1s_1a_2\cdots \in(\states\actions)^\omega$ such that for every $n\in\nat$, $a_{n+1}\in\actions(s_n)$. A \emph{finite play} in $\arena$ is a finite sequence in $\states(\actions\states)^*$ which is the prefix of an infinite play.

With each infinite play is associated a payoff computed by a \emph{payoff function}.  Player~1 (the maximizer) wants to maximize the expected payoff while Player~2 (the minimizer) has the exact opposite preference. Formally, a payoff function for the arena $\arena$ is a bounded and Borel-measurable function
\begin{align*}
 \pay:(\states\actions)^\omega \to \rel 
\end{align*}
which associates with each infinite play $h$ a payoff $\pay(h)$.

\begin{definition}[Stochastic game with perfect information]
  A stochastic game with perfect information is a pair
  \begin{align*}
    (\arena,\pay)
  \end{align*}
  where $\arena$ is an arena and $\pay$ a payoff function for the arena $\arena$.
\end{definition}

\paragraph{Strategies}

A \emph{strategy} in an arena $\arena$ for Player 1  is a function
\begin{align*}
  \sigma \st (\states\actions)^*\states_1 \to \distrib{\actions} 
\end{align*}
such that for any finite play $s_0a_1\cdots s_n$, and every action $a\in\actions$, if $\sigma(s_0a_1\cdots s_n)(a) > 0$ then the action $a$ belongs to $\actions(s_n)$, \ie the played action is available. Strategies for Player~2 are defined similarly and are typically denoted $\tau$. General strategies can have infinite memory as well as randomise among the available actions at every step. We are interested in a very simple sub-class of strategies, namely those that do not use any memory, or randomisation. 

\begin{definition}[Deterministic and stationary strategies]
  \label{def:strats}
  A strategy $\sigma$ for Player 1 is \emph{deterministic} if for every finite play $h\in(\states\actions)^*\states_1$ and action $a\in\actions$,
\begin{align*}
  \sigma(h)(a)>0 \qquad \Leftrightarrow\qquad \sigma(h)(a)=1.
\end{align*}
A strategy $\sigma$ is \emph{stationary} if $\sigma(h)$ only depends on the last state of $h$. In other words $\sigma$ is stationary if for every state $t\in \states_1$ and for every finite play $h=s_0a_1\cdots a_k t$,
\begin{align*}
  \sigma(h)=\sigma(t).  
\end{align*}
\end{definition}

Given an initial state $s\in\states$ and strategies $\sigma$ and $\tau$ for players $1$ and $2$ respectively, the set of infinite plays that start at state $s$ is naturally equipped with a sigma-field and a probability measure denoted $\ppdd$ that are defined as follows.  Given a finite play $h$ and an action $a$, the set of infinite plays $h(\actions\states)^\omega$ and $ha(\states\actions)^\omega$ are \emph{cylinders} that we abusively denote $h$ and $ha$.  The sigma-field is the one generated by cylinders and $\ppdd$ is the unique probability measure on the set of infinite plays that start at $s$ such that for every finite play $h$ that ends in state $t$, for every action $a\in\actions$ and state $r\in\states$,
\begin{align}
%  \label{eq:decal1}
  \ppd{ha \mid h} &=
                     \begin{cases}
                       \sigma(h)(a)& \text{ if $t\in \states_1$},\\
                       \tau(h)(a) & \text{ if $t\in \states_2$},\\
                     \end{cases}\\
%  \label{eq:decal2}
  \ppd{har \mid ha} &=\psmall{t,a}{r}.
\end{align}
For $n\in\nat$, we denote $S_n$ and $A_n$ the random variables defined by
\begin{align*}
  S_n(s_0a_1s_1\cdots) &\defeq s_n,\\
  A_n(s_0a_1s_1\cdots) &\defeq a_n.
\end{align*}

\paragraph{Values and optimal strategies}
Let $\game$ be a game with a bounded measurable payoff function $\pay$. The expected payoff associated with an initial state $s$ and two strategies $\sigma$ and $\tau$ is the expected value of $\pay$ under $\ppdd$, denoted $\eed\pay$. The \emph{maxmin} and \emph{minmax} values of a state $s\in\states$ in the game $\game$ are:
\begin{align*}
\maxmin(\game)(s) &\defeq \sup_\sigma \inf_\tau \eed\pay,\\
\minmax(\game)(s) &\defeq \inf_\tau\sup_\sigma \eed\pay.
\end{align*}

By definition of $\maxmin$ and $\minmax$, for every state $s\in\states$, $\maxmin(\game)(s) \leq  \minmax(\game)(s)$. As a corollary of the Martin's determinacy theorem for Blackwell games~\cite[Section 1]{martin}, the converse inequality holds as well:
\begin{theorem}[Martin's second determinacy theorem, {\cite[Section 1]{martin}}]
  \label{thm:martin}
Let $\game$ be a game with a Borel-measurable and bounded payoff function $\pay$.
Then for every state $s\in\states$:
\begin{align*}
  \val(\game)(s)\defeq\maxmin(\game)(s)= \minmax(\game)(s).
\end{align*}
This common value is called the value of state $s$ in the game $\game$ and denoted $\val(\game)(s)$.
\end{theorem}
The existence of a value guarantees the existence of $\epsilon$-optimal strategies for both players and every $\epsilon>0$.
\begin{definition}[Optimal and $\epsilon$-optimal strategies]
Let $\game$ be a game, $\epsilon>0$ and $\sigma$ a strategy for Player 1.
Then $\sigma$ is $\epsilon$-optimal
if for every strategy $\tau$
and every state $s\in\states$,
\begin{align*}
  \ee s {\sigma,\tau} \pay  \geq \minmax(\game)(s) - \epsilon.
\end{align*}
The definition for Player~2 is symmetric.  A $0$-optimal strategy is simply called \emph{optimal}.
\end{definition}
A stronger class of $\epsilon$-optimal strategies are $\epsilon$-subgame-perfect strategies, which are strategies that are not only $\epsilon$-optimal from the initial state $s$ but stay $\epsilon$-optimal throughout the game. More precisely, given a finite play $h=s_0\cdots s_n$ and a function $g$ whose domain is the set of (in)finite plays, by $g[h]$ we denote the function $g$ \emph{shifted} by~$h$:
\begin{align*}
  g[h](t_0a_1t_1\cdots)\defeq
  \begin{cases}
    g(h a_1t_1\cdots) &\text{ if $s_n=t_0$},\\
    g(t_0a_1t_1\cdots) &\text{ otherwise.}
  \end{cases}
\end{align*}
\begin{definition}[$\epsilon$-Subgame-Perfect Strategy]
  \label{def:subgameperfect}
  Let $\game$ be a game equipped with a payoff function $\pay$. A strategy $\hat\sigma$ for Player~1 is said to be $\epsilon$-subgame-perfect if for every finite play $h:=s_0\cdots s_n$,
  \begin{align*}
    \inf_\tau\ee {s_n} {\hat\sigma[h],\tau} {f[h]}\ge \sup_{\sigma}\inf_\tau\ee{s_n}{\sigma,\tau}{f[h]}-\epsilon. 
  \end{align*}
\end{definition}

\paragraph{Shift-invariant and submixing}
Without loss of generality we can assume that there is a finite set $\colors$ (colours assigned to the states of the game) such that the payoff function $f$ is a function
\begin{align*}
  f\st \colors^\omega\to \rel,
\end{align*}
that is Borel-measurable and bounded. We define the two conditions with respect to such payoff functions.

\begin{definition}[Shift-Invariant]
  \label{def:shift invariant}
  The payoff function $f$ is shift-invariant if and only if for all finite prefixes $p\in \colors^*$ and trajectories $u\in\colors^\omega$,
  \begin{align*}
    f(p\ u)=f(u). 
  \end{align*}
\end{definition}
Note that shift-invariance is a stronger condition than saying: if one can get $u'\in\colors^\omega$ from $u\in\colors^\omega$ by replacing finitely many letters then $f(u)=f(u')$. Sometimes in the literature this stronger condition is called ``prefix-independent'' or ``tail-measurable''. Intuitively shift-invariant payoff functions are such that they only measure asymptotic properties, and do not talk about indices.

A factorisation of $u\in\colors^\omega$ is a sequence $u_1,u_2,\ldots$ of non-empty finite words (\ie elements of $\colors^+$) such that
\begin{align*}
  u=u_1u_2u_3\cdots.
\end{align*}
For ${\color{blue}u},{\color{red}v}, w\in\colors^\omega$, we say that $w$ is a shuffle of ${\color{blue}u}$ and ${\color{red}v}$ if there are respective factorisations ${\color{blue} u_1},{\color{blue} u_2},\ldots$, and ${\color{red} v_1},{\color{red} v_2},\ldots$ such that
\begin{align*}
  w={\color{blue} u_1}\ {\color{red} v_1}\ {\color{blue} u_2}\ {\color{red} v_2}\cdots .
\end{align*}
\begin{definition}[Submixing]
  The payoff function $f$ is submixing if and only if for all ${\color{blue} u}, {\color{red} v}, w\in\colors^\omega$ such that $w$ is a shuffle of ${\color{blue} u}$ and ${\color{red} v}$ we have
  \begin{align*}
    f(w)\le\max\set{f({\color{blue} u}), f({\color{red} v})}.
  \end{align*}
\end{definition}
The submixing condition says that one cannot shuffle two losing trajectories to make a winning one. This requirement simplifies the kind of strategies that the players need.

The submixing condition is not symmetric over the players, and it implies different results for different players (notice the difference between \Cref{thm:main} and \Cref{thm:memory}). We define the inverse-submixing condition which is its reflection about the players:
\begin{definition}[Inverse-Submixing]
  The payoff function $f$ is inverse-submixing if and only if for all ${\color{blue} u}, {\color{red} v}, w\in\colors^\omega$ such that $w$ is a shuffle of ${\color{blue} u}$ and ${\color{red} v}$ we have
  \begin{align*}
    f(w)\ge\min\set{f({\color{blue} u}), f({\color{red} v})}.
  \end{align*}
\end{definition}
There are payoff functions that are both submixing and inverse-submixing (\eg the parity function); for such payoffs \Cref{thm:main} implies simple optimal strategies for both players, \ie positionality.

\section{Applications and Examples}
\label{sec:applications}
In this section we give a variety of examples of payoff functions that are shift-invariant and submixing, some of them very well-known, others less so. Thus we unify a number of classical positional determinacy results and also sketch how straightforward it is to apply \Cref{thm:main} to novel payoff functions. Furthermore, we comment on the hypothesis of \Cref{thm:main}: Are the conditions necessary? What do they imply about the optimal strategies of the minimizer? Under what operations is this class of payoff functions closed? We start by listing a few well-known examples.

\subsection{Unification of Classical Results}
\label{sec:unification}

 The \emph{mean-payoff function} has been introduced by Gilette~\cite{gilette}.  It measures average performances. Each state $s\in\states$ is labeled with an immediate reward $r(s)\in\rel$. With an infinite play $s_0a_1s_1\cdots$ is associated an infinite sequence of rewards $r_0=r(s_0),r_1=r(s_1),\ldots$
and the payoff is:
\begin{align*}
\mean(r_0r_1\cdots) \defeq \limsup_{n}\frac{1}{n+1}\sum_{i=0}^n r_i.
\end{align*}

The \emph{discounted payoff} has been introduced by Shapley~\cite{shapley}.  It measures long-term performances with an inflation rate: immediate rewards are discounted.  Each state $s$ is labeled not only with an immediate reward $r(s)\in\rel$ but also with a discount factor $0\le \lambda(s) <1$.  With an infinite play $h$ labeled with the sequence $(r_0,\lambda_0)(r_1,\lambda_1)\cdots\in(\rel\times[0,1))^\omega$ of daily payoffs and discount factors is associated the payoff:
\begin{align*}
  \mdisc\left((r_0,\lambda_0)(r_1,\lambda_1)\cdots\right) \defeq r_0 + \lambda_0 r_1 + \lambda_0\lambda_1 r_2 +\cdots.
\end{align*}

The \emph{parity condition} is used in automata theory and logics~\cite{dagstuhl}. Each state $s$ is labeled with some color $c(s)\in\set{0,\ldots,d}$.  The payoff is $1$ if the highest color seen infinitely often is even, and $0$ otherwise.  For $c_0c_1\cdots\in\set{0,\ldots,d}^\omega$,
\begin{align*}
  \parite(c_0c_1\cdots)\defeq
  \begin{cases}
    0 \text{ if }\limsup_{n} c_n \text{ is even,}\\
    1 \text{ otherwise.}
  \end{cases}
\end{align*}

The \emph{limsup payoff function} has been used in the theory of gambling games~\cite{maitrasudderth}.  States are labeled with immediate rewards and the payoff is the limit supremum of the rewards:
\begin{align*}
 \paylimsup(r_0r_1\cdots) \defeq \limsup_n r_n.
\end{align*}
The \emph{liminf payoff function} can be defined similarly. 

The two following propositions follow easily from \Cref{thm:main}. 
\begin{proposition}
  \label{prop:exemplesconcavite}
  The payoff functions $\paylimsup$, $\payliminf$, $\parite$ and $\mean$ are shift-invariant and submixing. Moreover $\paylimsup$, $\payliminf$, and $\parite$ are inverse-submixing as well.
\end{proposition}
\begin{proposition}
  \label{cor:unif}
  In every two-player stochastic game equipped with the parity, limsup, liminf, mean or discounted payoff function, Player~1 has a deterministic and stationary strategy which is optimal. The same is true for Player~2 for the parity, limsup and liminf payoff. 
\end{proposition}
One comment should be made about the discounted payoff function: While it is not shift-invariant, it is possible to reduce games equipped with this function to games with the mean-payoff function, by interpreting discount factors as stopping probabilities as was done in the seminal paper of Shapley~\cite{shapley}.  One can find details of this reduction in~\cite{1jpos,mathese}.

Thus we have unified a number of classical results, thereby giving a common reason for the half-positionality of seemingly unrelated games. The approaches that can be found in the literature for proving that these games are (half-)positional are diverse, as one can see, for example, by consulting the papers \cite{CourYan:1990} and \cite{maitrasudderth} that show positionality for parity games and limsup games, respectively.  The existence of deterministic and stationary optimal strategies in mean-payoff games has a colourful history attached.  The first proof was given by Gilette~\cite{gilette} based on a variant of Hardy and Littlewood theorem.  Later on, Ligget and Lippman found the variant to be wrong and proposed an alternative proof based on the existence of Blackwell optimal strategies plus a uniform boundedness result of Brown~\cite{liggettlippman}.  For one-player games, Bierth~\cite{bierth} gave a proof using martingales and elementary linear algebra while~\cite{lp} provided a proof based on linear programming and a modern proof can be found in~\cite{stochasticgames} based on a reduction to discounted games and the use of analytical tools.  For two-player games, a proof based on a transfer theorem from one-player to two-player games can be found in~\cite{mathese,gimbert:hal-00438359,gimbert:hal-004383592}.

\subsection{Other Examples}

We mention a few more recent examples of games. 

One-counter stochastic games have been introduced in~\cite{onecountergames},
in these games each state $s\in \states$
is labeled by a relative integer $c(s)\in \mathbb{Z}$.
Three different winning conditions were defined and studied in~\cite{onecountergames}:
\begin{align}
  \label{def:limsuppos}&\limsup_n \sum_{0\leq i \leq n} c_i = + \infty\\ 
  \label{def:liminfneg}&\limsup_n \sum_{0\leq i \leq n} c_i = - \infty\\
  \label{def:posavg}&\mean (c_0c_1\ldots)  > 0
\end{align}
The \emph{positive average condition} defined by~\eqref{def:posavg} is a variant of mean-payoff payoff, which may be more suitable to model quality of service constraints or decision makers with a loss aversion. One can naturally defined a payoff function $\payposavg$, that outputs $1$ if the condition holds, and $0$ otherwise. 

Although $\payposavg$ seems similar to the $\mean$ function, maximizing the expected value of $\payposavg$ and doing the same for $\mean$, are two different goals. For example, a positive average maximizer prefers seeing the sequence $1,1,1,\ldots$ for sure rather than seeing with equal probability $\frac{1}{2}$ the sequences $0,0,0,\ldots$ or $3,3,3,\ldots$ while a mean-value maximizer prefers the second situation to the first one.
To the best knowledge of the authors, the classical techniques developed  in~\cite{bierth,stochasticgames,lp} cannot be used to prove positionality of games equipped with the positive average condition. However, since $\payposavg$ can be defined as the composition of the submixing function $\mean$ with an increasing function it is submixing itself. As a consequence of the main theorem of the present paper, it then follows that games that are equpped with $\payposavg$ are half-positional.

Another recent example are the \emph{generalized mean payoff games}, that were introduced in~\cite{generalizedmeanpayoff}.  Each state is labeled by a fixed number of immediate rewards $\left(r^{(1)},\ldots,r^{(k)}\right)$, which define as many mean payoff conditions $\left(\mean^1,\ldots, \mean^k\right)$.  The winning condition is:
\begin{align}
  \label{eq:generalizedmeanpayof}
  \forall 1\leq i \leq k, \mean^i\left(r^{(i)}_0r^{(i)}_1\ldots\right) > 0.
\end{align}
In the special case of mean-payoff co-B\"uchi games, a subset of the states are called B\"uchi states, and the payoff of Player 1 is $-\infty$ if B\"uchi states are visited infinitely often and the mean-payoff value of the rewards otherwise. One can easily check that such a payoff mapping is shift-invariant and submixing.  Although we do not explicitly handle payoff mappings that take infinite values, it is possible to approximate the payoff function by replacing $-\infty$ by arbitrary small values to prove half-positionality of mean-payoff co-B\"uchi games.

The general payoffs captured by the condition in~\eqref{eq:generalizedmeanpayof} are not submixing, however, a natural variant is: 
\emph{Optimistic generalized mean-payoff games} are defined similarly except the winning condition is
\begin{align*}
  \exists i, \mean^{i} \geq 0.
\end{align*}
It is an exercise to show that this winning condition is submixing.  More generally, if $f_1,\ldots, f_n$ are submixing payoff mappings then $\max\{ f_1,\ldots,f_n \}$ is submixing as well.  As a consequence of this observation and \Cref{thm:main}, games with the optimistic generalized mean-payoff condition are half-positional. Such games are not positional however. One can show that the minimizer requires (finite) memory. Intuitively, he needs to use the memory to remember which dimensions have to be decreased, in order to render the condition false. There are even examples of shift-invariant and submixing payoff functions where the minimizer requires infinite memory to play optimally. Here is one of them.

The set of colours is $\set{a,b}$. The payoff function is equal to $-1$ if and only if the word $w\in\set{a,b}^\omega$ that it inputs contains infinitely many $a$s, infinitely many $b$s, and moreover
\begin{align*}
  w=a^{n_1}ba^{n_2}ba^{n_3}b\cdots,
\end{align*}
is such that $\liminf_\ell n_\ell=\infty$, otherwise it is equal to $0$. 

One final but interesting example of a payoff function that is shift-invariant, submixing, and even inverse-submixing (hence positional for both players in two-players games) is the \emph{positive frequency payoff}. Every state is labeled by a color from a set $C$, each of which has a payoff $u(c)$. An infinite play generates an infinite word of colors:
\begin{align*}
  w\defeq c_0c_1c_2\cdots,
\end{align*}
For a color $c$ and $n\in\nat$ define $\#(c,c_0c_1\cdots c_n)$ to be the number of occurrences of the color $c$ in the prefix $c_0c_1\cdots c_n$. The frequency of the color $c$ in $w$ is defined as:
\begin{align*}
  \mathrm{freq}(c,w)\defeq \limsup_{n\to\infty} \frac{\#(c,c_0c_1\cdots c_n)} n,
\end{align*}
and the payoff
\begin{align*}
  f_{\mathrm{freq}}(w)\defeq \max\set{u(c)\st c\in C,\  \mathrm{freq}(c,w)>0}. 
\end{align*}
Other examples can be found in~\cite{1jpos, thesekop, mathese}, and in the papers cited in the introduction.

\subsection{The Class of Shift-Invariant and Submixing Functions}
In this section we have already used two operators under which the class of shift-invariant and submixing functions is closed:
\begin{itemize}
\item If $f_1,\ldots, f_k$ are shift-invariant and submixing then so is
  \begin{align*}
    f(w)\defeq\max\set{f_1(w),\ldots,f_k(w)}.
  \end{align*}
\item If $f$ is shift-invariant and submixing, and $g$ is an increasing function then
  \begin{align*}
    g\circ f 
  \end{align*}
  is shift-invariant and submixing. 
\end{itemize}
The proofs are routine.

The class of shift-invariant and submixing functions does not seem to have any non-trivial closure property. For example, even though this class is closed under $\max$ above, it is not closed under addition. That is if $f_1$ and $f_2$ are submixing, then $f(w):=f_1(w)+f_2(w)$ need not be. To see this, consider the example with colors $a$ and $b$, and $f_1$ such that it maps to $1$ if $a$ occurs infinitely often, and $0$ otherwise, and $f_2$ defined symmetrically.

Furthermore, neither condition is necessary in~\Cref{thm:main}: discounted games are positional but not shift-invariant, and $\mean$ with $\liminf$ instead of $\limsup$ is positional but not submixing. However, as we have seen, this class contains many interesting payoff functions, and it is the salient property that allows one to prove the existence of positional optimal strategies. Perhaps even more importantly, it is typically trivial to check whether a given payoff function is shift-invariant and submixing. 

%%% Local Variables:
%%% mode: latex
%%% TeX-master: "main"
%%% End:

\section{$\epsilon$-Subgame-Perfect Strategies}
\label{sec:subgameperfect}

The proof of \Cref{thm:main} hinges on a crucial property of games with perfect information, namely the fact that they admit $\epsilon$-subgame-perfect strategies, for all $\epsilon>0$. 

\subgameperfect*

Note that we cannot lift the shift-invariant hypothesis from \Cref{thm:subgameperfect}. That is, one can easily find an example of a game where there are no $\epsilon$-subgame-perfect strategies, even a game with only one player. 

Note that \Cref{thm:subgameperfect}  is true for arbitrary payoff functions and the weaker notion of $\epsilon$-subgame-perfect strategy, requiring that for every finite play $h=s_0\cdots s_n$,
\begin{align}
  \label{eq:weaker subgameperfect}
  \inf_\tau\ee {s_n} {\sigma[h],\tau} {f[h]}\ge \sup_{\sigma'}\inf_\tau\ee{s_n}{\sigma',\tau}{f[h]}-\epsilon. 
\end{align}
Indeed, this was proved independently by Mashiah-Yaakovi, \cite[Proposition 11]{correlated} for concurrent games. That result implies \Cref{thm:subgameperfect}, since for shift-invariant games, the condition \eqref{eq:weaker subgameperfect} coincides with that of \Cref{def:subgameperfect}. 

The weak and strong notions of conditions coincide when the payoff function is shift-invariant. On the one hand, our proof only works for the strong notion of subgame-perfectness in \Cref{def:subgameperfect}. On the other hand, our proof makes transparent how to construct $\epsilon$-subgame-perfect strategies from $\epsilon/2$-optimal ones, in a way that preserves some important properties of the strategy, notably its use of finite memory.

The proof of the theorem will be symmetric with respect to the players, so we will only show that Player~1 has $\epsilon$-subgame-perfect strategies. We will do this by taking an $\epsilon$-optimal strategy $\sigma$ with some more structure, and using it to construct a \emph{reset} strategy $\reset\sigma$, which will be $2\epsilon$-subgame-perfect. The reset strategy is conceptually very simple: a strategy $\sigma$ is not $2\epsilon$-subgame-perfect if and only if there exists some finite play $h:=s_0\cdots s_n$ such that
\begin{align}
  \label{eq:drop}
  \inf_\tau\ee {s_n}{\sigma[h],\tau}{f[h]} < \val(s_n)-2\epsilon;
\end{align}
the reset strategy simply resets its memory when this happens. We give the formal definitions.
\begin{definition}
  \label{def:drop}The finite play $h:=s_0\cdots s_n$ is called a \emph{$(\epsilon,\sigma)$-drop} if \eqref{eq:drop} holds. We write
  \begin{align*}
    \Delta(\epsilon,\sigma)(h)\qquad\Leftrightarrow\qquad\text{h is a $(\epsilon,\sigma)$-drop}.
  \end{align*}
\end{definition}
It is plain that one can factorise any infinite play into $h_1 h_2\cdots$ where each $h_i$ is a $(\epsilon,\sigma)$-drop, but no strict prefix of $h_i$ is $(\epsilon,\sigma)$-drop. For example:
\begin{center}
\includegraphics[width=.8\textwidth]{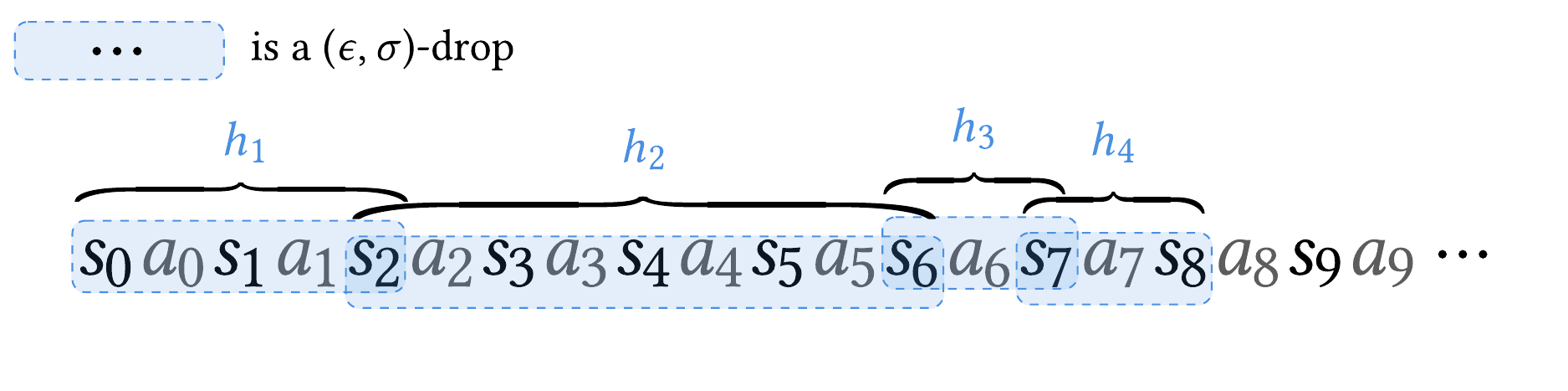}
\end{center}
\begin{definition}
  We define the \emph{date} of the most recent (or latest) drop for all $s_0\cdots s_n$ inductively as:
  \begin{align*}
    \Lambda(\epsilon,\sigma)(s_0)&\defeq 0\\
    \Lambda(\epsilon,\sigma)(s_0\cdots s_n)&\defeq
    \begin{cases}
      n \qquad &\text{if $h$ is a $(\epsilon,\sigma)$-drop}\\
      \Lambda(\epsilon,\sigma)(s_0\cdots s_{n-1})\qquad &\text{otherwise},
    \end{cases}
  \end{align*}
  where
  \begin{align*}
    h\defeq s_{\ell}\cdots s_n,\ \ \ \text{and} \ \ \ \ell \defeq \Lambda(\epsilon,\sigma)(s_0\cdots s_{n-1}).
  \end{align*}
\end{definition}
The date of the most recent drop in the example above looks as follows:
\begin{center}
\includegraphics[width=.8\textwidth]{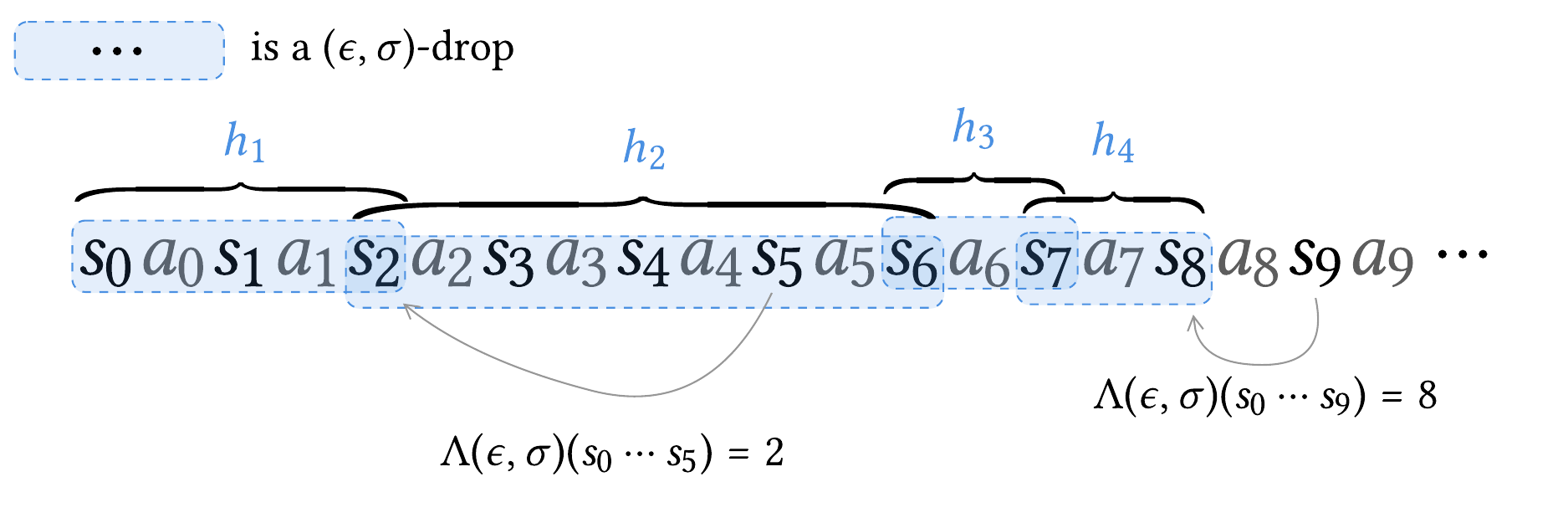}
\end{center}

The reset strategy resets its memory whenever a drop occurs, \ie it keeps the memory since the most recent drop:
\begin{definition}[Reset Strategy]
  For any strategy $\sigma$ we define the reset strategy $\reset\sigma$ as:
  \begin{align*}
    \reset\sigma(s_0\cdots s_n)=\sigma(s_\ell\cdots s_n),
  \end{align*}
  where
  \begin{align*}
    \ell\defeq\Lambda(\epsilon,\sigma)(s_0\cdots s_n).
  \end{align*}
\end{definition}

By construction, the reset strategy has the property that if it is $\epsilon$-optimal then it is also $2\epsilon$-subgame-perfect.

\begin{lemma}
  \label{lem:2e subgame}
  Let $\reset\sigma$ be a reset strategy that is $\epsilon$-optimal, then it is also $2\epsilon$-subgame-perfect. 
\end{lemma}
\begin{proof}
  Let $s_0\cdots s_n$ be a finite play, the goal is to show that:
  \begin{align}
    \label{eq:goal 2e}
    \inf_\tau\ee{s_n}{\reset\sigma[s_0\cdots s_n],\tau} f\ge \val(s_n)-2\epsilon. 
  \end{align}
  If there is a drop occurring in date $n$, that is $\Lambda(\epsilon,\sigma)(s_0\cdots s_n)=n$ then
  \begin{align*}
    \inf_\tau\ee{s_n}{\reset\sigma[s_0\cdots s_n],\tau} f=\inf_\tau\ee{s_n}{\reset\sigma,\tau} f\ge\val(s_n)-\epsilon,
  \end{align*}
  by the definition of a reset strategy that is $\epsilon$-optimal. Assume then that the most recent drop is $\ell<n$, which means that:
  \begin{align}
    \label{eq:to contradict 2e}
    \inf_\tau\ee {s_n} {\sigma[s_\ell\cdots s_n],\tau} f \ge \val(s_n)-2\epsilon,
  \end{align}
  where $\ell=\Lambda(\epsilon,\sigma)(s_0\cdots s_{n-1})$. Towards a contradiction, assume that the goal \eqref{eq:goal 2e} does not hold, \ie there exists a strategy $\tau$ that gives payoff strictly less than $\val(s_n)-2\epsilon$, then we will construct another strategy $\tau'$ that will contradict \eqref{eq:to contradict 2e}.
  
Let $\mathfrak D$ be the set prefixes from $s_\ell$ to the next $(\epsilon,\sigma)$-drop, that is
  \begin{align*}
    \mathfrak D\defeq \set{s_\ell\cdots s_{\ell'}\st s_\ell\cdots s_{\ell'}\text{ is a $(\epsilon,\sigma)$-drop but no strict prefix is}},
  \end{align*}
  and $\overline{\mathfrak D}$ the event that is generated by the cylinders in $\mathfrak D$ (note that the complement $\neg\overline{\mathfrak D}$ is the event that no drop occurs). Define $\tau'$ to be the strategy that plays like $\tau$ except when a prefix in $\mathfrak D$ is met, in which case it switches to the $\epsilon$-response strategy $\tau''$. To simplify the notation let:
  \begin{align*}
    \sigma_1\defeq \reset\sigma[s_0\cdots s_n],\qquad \sigma_2\defeq \sigma[s_\ell\cdots s_n].
  \end{align*}
  From the assumption that the goal does not hold we have the following inequality\footnote{$\mathds{1}_\mathcal{E}$ is the indicator function of the event $\mathcal E$.}
  \begin{align}
    \label{eq:to contradict 2e'}
    \begin{split}
    \val(s_n)-2\epsilon &> \ee{s_n}{\sigma_1,\tau}{f\cdot\mathds{1}_{\overline{\mathfrak D}}}+\ee{s_n}{\sigma_1,\tau}{f\cdot\mathds{1}_{\neg\overline{\mathfrak D}}}\\
                        &=\ee{s_n}{\sigma_1,\tau}{f\cdot\mathds{1}_{\overline{\mathfrak D}}}+\ee{s_n}{\sigma_2,\tau'}{f\cdot\mathds{1}_{\neg\overline{\mathfrak D}}}\\
                        &=\ee{s_n}{\sigma_1,\tau}{f\cdot\mathds{1}_{\overline{\mathfrak D}}}+\ee{s_n}{\sigma_2,\tau'} f - \ee{s_n}{\sigma_2,\tau'}{f\cdot\mathds{1}_{\overline{\mathfrak D}}}.
    \end{split}
  \end{align}
  In the equality the strategy $\sigma_1$, respectively $\tau$, has been replaced by $\sigma_2$, respectively $\tau'$ because on infinite plays without a drop they coincide.

  For the first term above we have:
  \begin{align*}
    \ee{s_n}{\sigma_1,\tau}{f\cdot\mathds{1}_{\overline{\mathfrak D}}} &= \sum_{t_0\cdots t_m\in \mathfrak D}\pp {s_n} {\sigma_1,\tau} {t_0\cdots t_m} \ee {t_m} {\reset\sigma,\tau[t_0\cdots t_m]} f\\
    &\ge \sum_{t_0\cdots t_m\in \mathfrak D}\pp {s_n} {\sigma_1,\tau} {t_0\cdots t_m} (\val(t_m)-e), 
  \end{align*}
  by definition of the $\epsilon$-optimal reset strategy and the fact that $f$ is shift-invariant. For the last term on the other hand we have:
  \begin{align*}
    \ee{s_n}{\sigma_2,\tau'}{f\cdot\mathds{1}_{\overline{\mathfrak D}}} &= \sum_{t_0\cdots t_m\in\mathfrak D}\pp {s_n} {\sigma_2,\tau'} {t_0\cdots t_m} \ee{t_m}{\sigma_2[t_0\cdots t_m],\tau''} f\\
    &\le \sum_{t_0\cdots t_m\in\mathfrak D}\pp {s_n} {\sigma_2,\tau'} {t_0\cdots t_m} (\val(t_m)-2\epsilon+\epsilon),
  \end{align*}
  by construction of the $\epsilon$-response strategy $\tau''$ and $\tau'$. The strategies $\sigma_1$ and $\sigma_2$ on one hand, and $\tau$, $\tau'$ on the other, coincide up to the first drop, consequently we can interchange them when measuring cylinders $t_0\cdots t_m$, which implies that the two inequalities above give:
  \begin{align*}
    \ee {s_n} {\sigma_1,\tau}{f\cdot \mathds{1}_{\overline{\mathfrak D}}}\ge \ee {s_n} {\sigma_2,\tau'}{f\cdot \mathds{1}_{\overline{\mathfrak D}}}.
  \end{align*}
  This contradicts \eqref{eq:to contradict 2e} when plugged it in \eqref{eq:to contradict 2e'}. 
\end{proof}

As a consequence of this lemma, in order to prove \Cref{thm:subgameperfect}, we only have to demonstrate that there exists a reset strategy that is $\epsilon$-optimal. In the rest of this section we will prove that there are strategies with more and more desirable properties, culminating in the proof that there is some $\sigma$ whose reset strategy is $\epsilon$-optimal.

\subsection{Properties of the Reset Strategy}

We will show that there is a strategy $\sigma$ with the following properties:
\begin{enumerate}
\item $\sigma$ is $\epsilon$-optimal,
\item $\sigma$ is locally optimal,\footnote{This means that it does not play an action that decreases the value on average, the precise definition will follow. }
\item for any $\tau$ when playing with $\reset\sigma$ and $\tau$ almost surely there are only finitely many $(\epsilon,\sigma)$-drops, and
\item $\reset\sigma$ is $\epsilon$-optimal. 
\end{enumerate}
We will do this in a manner that accumulates more structure, that is, for strategies with properties 1 and 2 we can prove the third property; and for strategies with all of the first three properties it is possible to prove that the reset strategy is $\epsilon$-optimal. Each subsection below corresponds to the proof of one of the last three properties (Property 1 is a consequence of Martin's theorem \Cref{thm:martin}).

We are going to make use of some results from the theory of martingales\footnote{As a general reference for this area one might use \cite{wiliamspr}.}, which we introduce first.
\begin{definition}[Martingale]
  A sequence of real-valued random variables $X_0,X_1,\ldots$ is called a martingale if for all $n\in\nat$
  \begin{align*}
    \ee{}{}{|X_n|}<\infty,\ \  \text{and}\ \ \ee{}{}{X_{n+1}\ \vert\  X_1,\ldots,X_n}=X_n.
  \end{align*}
  It is called a supermartingale, respectively submartingale, if instead of the equality we have $\ge$, respectively $\le$. 
\end{definition}

In our case the sequence $\val(S_0),\val(S_1),\ldots$ under suitable strategies will be a supermartingale, which will allow us to use in particular the following results.
\begin{theorem}[Doob's Forward Convergence Theorem, {\cite[Theorem 11.5]{wiliamspr}}]
  \label{thm:doob convergence}
  Let $X_0,X_1,\ldots$ be a supermartingale such that the sequence $(\ee{}{}{|X_n|})_{n\in\nat}$ is bounded. Then almost surely the limit
  \begin{align*}
    \lim_{n\to\infty}X_n,
  \end{align*}
  exists and is finite.
\end{theorem}
It follows from the definition of martingales that for all $n\in\nat$, the expected value of $X_n$ is equal to the expected value of $X_0$. In other words, the process that is stopped at time $n$ is on average is equal to the process at time $0$. The next theorem from martingale theory that we will make use of, has an analogous statement, namely that the process stopped at some random time $T$ is on average equal to the process stopped at time zero. This theorem is known as Doob's optional stopping theorem. See for example Section 10.10 in \cite{wiliamspr}. We give a variant of this theorem.
\begin{definition}[Stopping Time]
  A random variable $T$ taking values in $\nat\cup\set\infty$ is called a \emph{stopping time} with respect to random variables $S_0,S_1,\ldots$ if the event $\set{T=n}$ for $n\in\nat$ is $(S_0,\ldots,S_n)$-measurable, meaning that it depends only on the random variables $S_0,\ldots,S_n$. 
\end{definition}
\begin{theorem}[Doob's Optional Stopping Theorem]
  \label{thm:doob optional stopping}
  Let $T$ be a stopping time with respect to the random variables $S_0,S_1,\ldots$ and $(X_n)_{n\in\nat}$ a uniformly bounded martingale such that for all $n\in\nat$, $X_n$ is $(S_0,\ldots,S_n)$-measurable. Define the random variable $X_T$ which represents the process stopped at time $T$ as:
  \begin{align*}
    X_T\defeq
    \begin{cases}
      X_n\qquad &\text{if $T$ is finite and equal to $n$,}\\
      \lim_{n\to\infty}X_n\qquad &\text{if $T=\infty$}.
    \end{cases}
  \end{align*}
  Then the expectation of $X_T$ is equal to that of $X_0$. Analogous statements hold for supermartingales and submartingales.
\end{theorem}
\begin{proof}
  The random variable $X_T$ is well-defined as a consequence of \Cref{thm:doob convergence}. For every $k\in\nat$ define:
  \begin{align*}
    Y_k\defeq X_{\min(T,k)}. 
  \end{align*}
  The process $(Y_k)_{k\in\nat}$ is a uniformly bounded martingale that converges almost-surely, as well. By definition of martingales for all $n\in\nat$
  \begin{align*}
    \ee{}{}{Y_n}=\ee{}{}{Y_0}=\ee{}{}{X_0}.
  \end{align*}
  Furthermore $(Y_k)_{k\in\nat}$ converges pointwise to $X_T$. One can now use Lebesgue's dominated convergence theorem (see for example \cite[Theorem 5.9]{wiliamspr}) to conclude that:
  \begin{align*}
    \ee{}{}{X_T}=\ee{}{}{X_0}. 
  \end{align*}
  When the process is a supermartingale or a submartingale one can write an analogous proof.
\end{proof}
\subsubsection{Locally Optimal}
An action is locally optimal if the average value of the successor states is equal to the value of the current state. Formally:
\begin{definition}[Locally Optimal Strategy]
  An action $a\in\actions(s)$ is called locally optimal if and only if
  \begin{align*}
    \val(s)=\sum_{t\in\states}\psmall{s,a}{t}\val(t).
  \end{align*}
  A strategy that only plays locally optimal actions is called locally optimal.
\end{definition}
The salient point is the following observation about the process $\val(S_0),\val(S_1),\ldots$ when players use locally optimal strategies.
\begin{observation}
  \label{ob:is a martingale}
  When Player~1 (respectively Player~2) uses a locally optimal strategy the process
  \begin{align*}
    \val(S_0),\val(S_1),\ldots
  \end{align*}
  is a supermartingale (respectively a submartingale). 
\end{observation}
This observation readily follows from the definition above and the fact that the values are bounded.

One can get away with playing solely locally optimal actions in games with perfect information. In other words, suppose that the action $a_0\in\actions(s_0)$ (say belonging to Player~1) in game $\game$ is \emph{not} locally optimal, and denote by $\game'$ the same game except that it does not have action $a_0$ in state $s_0$. We will prove that the values of those two games coincide; this then clearly implies that Player~1 has $\epsilon$-optimal strategies that are locally optimal as well. The analogue fact for Player~2 can be proved symmetrically.

Player~1 has less choice in $\game'$, so for every $s\in\states$
\begin{align*}
  \val(\game')(s)\le \val(\game)(s),
\end{align*}
hence we only have to prove the inverse inequality. Towards this end, we first prove that:
\begin{align}
  \label{eq:g' s0}
  \val(\game')(s_0)\ge \val(\game)(s_0). 
\end{align}
Let
\begin{align*}
  \delta\defeq \val(\game)(s_0)-\sum_{t\in\states}\psmall{s_0,a_0} t \val(\game)(t)>0,
\end{align*}
and $\tau$ the strategy that plays according to the strategy $\tau'$ that is $\epsilon$-optimal in $\game'$ --- as long as the opponent does not choose the action $a_0$, in which case it switches definitely to the strategy $\tau''$ which is $\delta/2$-optimal in $\game$. Let $\mathcal Z$ be the event that the action $a_0$ is never chosen, \ie
\begin{align*}
  \mathcal Z \defeq \set{\forall n\ \  S_n=s_0\Rightarrow A_{n+1}\ne a}.
\end{align*}
Then by construction of $\tau$, for all $\sigma$ and $s$:
\begin{align*}
  \eed{f\ \vert\ \mathcal Z} &\le \val(\game')(s)+\epsilon,\text{ and}\\
  \eed{f\ \vert\ \neg\mathcal Z} &\le \val(\game)(s_0)-\delta+\delta/2,
\end{align*}
whence it follows that for all $\sigma$, $s$ and $\epsilon>0$
\begin{align*}
  \eed{f} \le \max\set{\val(\game')(s)+\epsilon,\val(\game)(s_0)-\delta/2}. 
\end{align*}
Taking $s=s_0$ and the supremum over all $\sigma$ gives \eqref{eq:g' s0}.

Using \eqref{eq:g' s0}, we prove now that for all $s$
\begin{align}
  \label{eq:g' s}
  \val(\game')(s)\ge\val(\game)(s). 
\end{align}
Define $\mathcal S(\sigma)$ to be the event that the action $a_0$ is about to be played by strategy $\sigma$, that is
\begin{align*}
  \mathcal S(\sigma)\defeq \set{\exists n\ \ S_n=s_0\text{ and }\sigma(S_0\cdots S_n)(a_0)>0}. 
\end{align*}
Let $\epsilon>0$ and for any strategy $\sigma$, define $\tilde \sigma$ to be the strategy that plays like $\sigma$ unless the latter is about to play the action $a_0$ in $s_0$, in which case it switches to the strategy $\sigma'$ which is $\epsilon$-optimal in $\game'$. Set $\tau$ to be the strategy that plays according to some strategy $\tau'$ which is $\epsilon$-optimal in $\game'$ as long as the opponent does not play the action $a_0$, otherwise it switches to some strategy that is $\epsilon$-optimal in $\game$. By definitions of these strategies and \eqref{eq:g' s0} we have that for all $\sigma$ and $s$
\begin{align*}
  \eed{f\ \vert\ \mathcal S(\sigma)}&\le \val(\game)(s_0)+\epsilon=\val(\game')(s_0)+\epsilon, \text{and}\\
  \ee s {\tilde\sigma,\tau} {f \vert \mathcal S(\sigma)} &\ge \val(\game')(s_0),
\end{align*}
a combination of which gives us
\begin{align}
  \label{eq:switch}
  \eed {f\ \vert\ \mathcal S(\sigma)}\le \ee s {\tilde\sigma,\tau} {f\ \vert\ \mathcal S(\sigma)}+2\epsilon. 
\end{align}
The strategies $\sigma$ and $\tilde\sigma$ on one hand, and $\tau$ and $\tau'$ on the other, coincide up to the date when $\sigma$ is about to play the action $a_0$, as a consequence:
\begin{align*}
  P(\sigma,s)\defeq \ppd {\mathcal S(\sigma)} = \pp s {\tilde\sigma,\tau'}{\mathcal S (\sigma)}. 
\end{align*}
Now by construction of the strategies and \eqref{eq:switch}, for all $\sigma$ and $s$ we have 
\begin{align*}
  \eed f &= P(\sigma,s)\ \eed {f\ \vert\ \mathcal S(\sigma)} + \left(1-P(\sigma,s)\right)\ \eed {f\ \vert\ \neg\mathcal S(\sigma)}\\
         &\le P(\sigma,s)\ \left(\ee s {\tilde\sigma,\tau} {f\ \vert\ \mathcal S(\sigma)} + 2\epsilon\right) + \left(1-P(\sigma,s)\right)\ \eed {f\ \vert\ \neg\mathcal S(\sigma)}\\
         &= \ee s {\tilde\sigma,\tau} {f} + 2\epsilon\ P(\sigma,s)
         = \ee s {\tilde\sigma,\tau'} {f} + 2\epsilon\ P(\sigma,s)\\
         &\le \val(\game')(s)+\epsilon(2P(\sigma,s)+1). 
\end{align*}
Since this holds for any $\epsilon>0$, taking the supremum over all $\sigma$ proves \eqref{eq:g' s}.

We have thus proved that for all $\epsilon>0$, both players have strategies that are both
\begin{align}
  \label{eq:locally optimal}
  \text{locally optimal and $\epsilon$-optimal.}
\end{align}

We gather one more observation about games where at least one of the players utilises a locally optimal strategy. In this case, a stronger type of locally optimal action is the only one played infinitely many times.

\begin{definition}[Value-Conserving Action]
  An action $a\in\actions(s)$ is called value-conserving in $s$ if and only if for all $t\in\states$,
  \begin{align*}
    \psmall{s,a} t > 0\qquad \Rightarrow\qquad \val(s)=\val(t). 
  \end{align*}
\end{definition}

\begin{proposition}
  \label{prop:value-conserving}
  For all strategies $\sigma,\tau$ at least one of which is locally optimal and $s\in\states$ we have
  \begin{align*}
    \ppd {\text{for all but finitely many $n$, $A_n$ is value-conserving in $S_n$}}=1. 
  \end{align*}
\end{proposition}
\begin{proof}
  Fix $\sigma$ and $\tau$ and assume that $\sigma$ is locally optimal, the other case is symmetrical. Suppose that $a_0\in\actions(s_0)$ is not value-conserving. It suffices to prove that the event
  \begin{align*}
    \set{\text{for infinitely many $n$, $S_n=s_0$ and $A_n=a_0$}},
  \end{align*}
  has measure zero. Assume towards a contradiction that the event above has non-zero probability, then the event which says that for infinitely many $n$, we have $S_n=s_0$, $A_n=a_0$ and $S_{n+1}=t$ also has non-zero probability; where $t\in\states$ is a successor state of $s_0$ under $a_0$ that has value strictly smaller than that of $s_0$ (its existence is guaranteed because $a_0$ is not value-conserving). This means that there is non-zero probability that for infinitely many $n$,
  \begin{align*}
    \vert \val(S_n)-\val(S_{n+1})\vert \ge \val(s_0)-\val(t)>0,
  \end{align*}
  which contradicts \Cref{thm:doob convergence}, since $(\val(S_n)), n\in\nat$ is a supermartingale as per \Cref{ob:is a martingale}. 
\end{proof}
\subsubsection{Finitely Many Drops}
\label{sub:finitely many drops}
Recall that $\Delta(\epsilon,\sigma)(\cdot)$ characterises finite plays that are $(\epsilon,\sigma)$-drops. We informally refer to the event
\begin{align*}
  \text{for all $m>n$, }\neg\Delta(\epsilon,\sigma)(S_0\cdots S_m),
\end{align*}
as
\begin{align*}
  \text{no $(\epsilon,\sigma)$-drops after date $n$}.
\end{align*}
Similarly for events such as ``there is a $(\epsilon,\sigma)$-drop'' or ``two $(\epsilon,\sigma)$-drops after date~$n$''. Our goal is to prove that for a reset strategy that is based on a $\sigma$ that is both $\epsilon$-optimal and locally optimal (which exists because of \eqref{eq:locally optimal}) almost surely there will only be finitely many $(\epsilon,\sigma)$-drops. To this end fix a $\epsilon>0$, and $\sigma$ a strategy that is both locally optimal and $\epsilon$-optimal, which allows us to simply say drop instead of $(\epsilon,\sigma)$-drop. The proof of the goal is relatively lengthy, however the idea and the plan is simple.

An intermediate fact that we have to prove is that when Player~1 plays with the reset strategy there is some $n\in\nat$ such that the probability that there is a drop after date $n$ is bounded away from $1$. This fact is easier to prove if we assume that the adversary is using a locally optimal strategy. Then \Cref{prop:value-conserving} helps us lift this restriction on the strategies of Player~2. Therefore the plan is to prove this intermediate fact first (1) for locally optimal strategies, then (2) for strategies $\tau_n$ that are locally optimal after date $n$, and finally (3) for general strategies. 
The intermediate fact then finalises the goal of the preset section, that is when Player~1 plays with the reset strategy $\reset\sigma$ almost surely there will be only finitely many drops. 
\begin{lemma}
  \label{lem:bounded away 1}
  There exists a $c>0$ such that for all $s$ and locally optimal $\tau$,
  \begin{align*}
    \ppd {\text{there is a drop}}\le 1-c.
  \end{align*}
\end{lemma}
\begin{proof}
  Let $T$ be the date of the first drop, that is
  \begin{align*}
    T\defeq \min\set{n\st S_0\cdots S_n\text{ is a drop }}, 
  \end{align*}
  with the convention that $\min\emptyset=\infty$. Notice that $T$ is a stopping time with respect to the process $(\val(S_n)), n\in\nat$. Let $\tau'$ be a strategy that plays like $\tau$ as long as no drop occurs, and once it does it switches to the strategy $\tau''$ that is a $\epsilon/2$-optimal response. By construction, $\tau$ and $\tau'$ coincide on trajectories without drops so define:
  \begin{align*}
    P\defeq \ppd{\text{no drops}} = \pp s {\sigma,\tau'} {\text{no drops}},
  \end{align*}
  and let $M$ respectively $m$, be an upper bound, respectively lower bound of the payoff function $f$. By $\epsilon$-optimality of $\sigma$, for all $s$ we have:
  \begin{align*}
    \val(s)-\epsilon &\le (1-P)\cdot \ee s {\sigma,\tau'} {f\ \vert\ \text{there is a drop}} + P\cdot \ee s {\sigma,\tau'} {f\ \vert\ \text{no drops}}\\
    &\le (1-P)\cdot \ee s {\sigma,\tau'} {f\ \vert\ \text{there is a drop}} + P\cdot M.
  \end{align*}
  Denote by $\mathcal D$ the finite plays that are drops but that do not have a prefix that is a drop, \ie it contains all the finite plays up to the first drop. Then by construction of $\tau'$ we have for all $s$:
  \begin{align*}
    \ee s {\sigma,\tau'} {f\ \vert\ \text{there is a drop}} &= \sum_{s_0\cdots s_n\in\mathcal D}\pp s {\sigma,\tau'} {s_0\cdots s_n\ \vert\ \text{there is a drop}}\cdot \ee {s_n} {\sigma[s_0\cdots s_n],\tau''} f\\
                                                            &\le\sum_{s_0\cdots s_n\in\mathcal D}\pp s {\sigma,\tau'} {s_0\cdots s_n\ \vert\ \text{there is a drop}}\cdot\left(\val(s_n)-2\epsilon+\epsilon/2\right)\\
    &=\ee s {\sigma,\tau'} {\val(S_T)\ \vert\ \text{there is a drop}}-\frac 3 2 \epsilon.
  \end{align*}
  Replacing this inequality in the one above and decomposing the expectation of $\val(S_T)$\footnote{~\Cref{thm:doob convergence} implies that this random variable is well-defined.} we conclude that for all $s$:
  \begin{align*}
    \val(s)-\epsilon &\le \ee s {\sigma,\tau'} {\val(S_T)} + P\cdot \left(M - \ee s {\sigma,\tau'} {\val(S_T)\ \vert\ \text{no drops}}\right)-\frac 3 2 \epsilon(1-P)\\
    &\le \val(s)+P\cdot (M-m) - \frac 3 2 \epsilon(1-P),
  \end{align*}
  where the expectation of $\val(S_T)$ is smaller than the $\val(s)$ for the following reason. Since $T$ is a stopping time and $\tau'$ plays like $\tau$ before the first drop, hence it plays locally optimal actions, consequently the process $\val(S_n), n\in\nat$ is a submartingale at least until the first drop\footnote{formally one defines another process that \emph{stops} after time T, that is a process $\val(S_{\min\set{n,T}})$. }, so we can apply \Cref{thm:doob optional stopping}. Finally from the inequality above we have:
  \begin{align*}
    P\ge \frac 1 2 \frac \epsilon {M-m +3/2\epsilon} \defeq c,
  \end{align*}
  a uniform bound that does not depend on the choice of $\tau$.
\end{proof}

Next we approximate strategies $\tau$ by a sequence $\tau_n$ for every natural $n$ as follows. The strategies $\tau_n$ play like $\tau$ only up to date $n$, otherwise they choose some locally optimal action, formally:
\begin{align*}
  \tau_n(s_0\cdots s_m) \defeq \begin{cases}
    \tau(s_0\cdots s_m)\ \text{if $m<n$ or $\tau(s_0\cdots s_m)$ chooses locally optimal actions},\\
    \text{some locally optimal action in $s_m$ otherwise.}
  \end{cases}
\end{align*}
\begin{lemma}
  \label{lem:tau n}
  There is some $c>0$ such that for all strategies $\tau$, $s$ and $n\in\nat$, we have
  \begin{align*}
    \pp s {\reset\sigma,\tau_n} {\text{there is a drop after date $n$}}\le 1-c. 
  \end{align*}
\end{lemma}
\begin{proof}
  For $n\in\nat$ define the stopping time $T_n$ to be the date of the first drop after the date $n$, that is
  \begin{align*}
    T_n\defeq \min\set{m > n\st S_0\cdots S_m \text{ is a drop}},
  \end{align*}
  with the convention that $\min\emptyset=\infty$, and set $T^2_n$ to be the date of the second drop after $n$, that is $T_{T_n}$. We prove that there is some $c>0$ such that for all $n\in\nat$, strategy $\tau$ and state $s$ we have
  \begin{align}
    \label{eq:two drops}
    \pp {s} {\reset\sigma,\tau_n} {T_n^2<\infty\ \vert\ T_n<\infty}\le 1-c.
  \end{align}
  The statement of the lemma then follows from \eqref{eq:two drops} and sigma-additivity of measures. Define $\mathcal D_n$ to be the set of finite plays, strictly longer than $n$, that are drops but such that they have no prefix longer than $n$ that is a drop. In other words $\mathcal D_n$ contains all the plays up to the first drop after the date $n$. Then by construction of the reset strategy:
  \begin{align*}
    \pp {s} {\reset\sigma,\tau_n} {T_n^2<\infty\ \vert\ T_n<\infty} &= \sum_{s_0\cdots s_m\in \mathcal D_n}\pp s {\reset\sigma,\tau_n} {T_n^2<\infty\ \vert\ s_0\cdots s_m} \pp s {\reset\sigma,\tau_n} {s_0\cdots s_m\ \vert\ T_n<\infty}\\
   &=\sum_{s_0\cdots s_m\in \mathcal D_n}\pp {s_m} {\reset\sigma,\tau_n[s_0\cdots s_m]} {T_0<\infty} \pp s {\reset\sigma,\tau_n} {s_0\cdots s_m\ \vert\ T_n<\infty}\\
&=\sum_{s_0\cdots s_m\in \mathcal D_n}\pp {s_m} {\sigma,\tau_n[s_0\cdots s_m]} {T_0<\infty} \pp s {\reset\sigma,\tau_n} {s_0\cdots s_m\ \vert\ T_n<\infty},
  \end{align*}
  where in the last equality we have replaced the reset strategy by $\sigma$, because these two strategies are the same up to the first drop. Since $m>n$, by construction the strategy $\tau_n[s_0\cdots s_m]$ is locally optimal, consequently applying \Cref{lem:bounded away 1} gives 
  \begin{align*}
    \pp {s_m} {\sigma,\tau_n[s_0\cdots s_m]} {T_0<\infty}\le 1-c,
  \end{align*}
  which when plugged into the equation above proves \eqref{eq:two drops}. 
\end{proof}

In the third lemma there is no restriction upon the strategy $\tau$.

\begin{lemma}
  \label{lem:general <1}
  For all strategies $\tau$ and $s$ there is some $n\in\nat$ such that
  \begin{align*}
    \pp s {\reset\sigma,\tau}{\text{there is a drop after date $n$}}<1. 
  \end{align*}
\end{lemma}
\begin{proof}
  Fix a strategy $\tau$ and a state $s$. Let $T$ be the stopping time that gives the date of the last action that was played that is not value-conserving, if it exists, otherwise let it be $\infty$. Since the strategies $\tau$ and $\tau_n$ coincide on all paths where the last action that is not value-conserving is played before $n$ (that is on the event $T<n$), then for all $n\in\nat$ and events $\mathcal E$ we have:
  \begin{align*}
    \pp s {\reset\sigma,\tau}{\mathcal E}=\pp s {\reset\sigma,\tau} {T<n}\cdot \pp s {\reset\sigma,\tau_n} {\mathcal E\ \vert\ T<n}+\pp s {\reset\sigma,\tau}{T\ge n}\cdot \pp s {\reset\sigma,\tau} {\mathcal E\ \vert\ T\ge n}.
  \end{align*}
  The strategy $\sigma$ has been assumed to be locally optimal, and therefore the strategy $\reset\sigma$ is locally optimal as well. As a consequence of \Cref{prop:value-conserving} we have
  \begin{align*}
    \lim_{n\to\infty} \pp s {\reset\sigma,\tau} {T<n}=1,
  \end{align*}
  whence follows
  \begin{align*}
    \lim_{n\to\infty}\pp s {\reset\sigma,\tau_n} {\mathcal E}= \pp s {\reset\sigma,\tau}{\mathcal E},
  \end{align*}
  for any event $\mathcal E$. The proof of the lemma now concludes by choosing the event ``there is a drop after date $n$'' for $\mathcal E$, a suitable natural number $n$ and applying \Cref{lem:tau n}. 
\end{proof}

This lemma makes it possible now to prove the third property of the strategy $\sigma$, namely that for all strategies $\tau$ and $s$,
\begin{align}
  \label{eq:third property}
  \pp s {\reset\sigma,\tau} {\exists n\ \ \text{no drops after date $n$}} = 1. 
\end{align}
Let $T$ be the stopping time that gives the date of the last drop, if it exists otherwise let it be equal to $\infty$. For a natural $n$, let $F_n$ be the stopping time that gives the date of the first drop after $n$ (same as $T_n$ in the proof of \Cref{lem:tau n}) if it exists, otherwise say that it is equal to $\infty$. 

Fix $\delta>0$ and choose the strategy $\tilde\tau$ and state $\tilde s$ such that
\begin{align}
  \label{eq:sup}
  \sup_{\tau,s}\pp {s} {\reset\sigma,\tau} {T=\infty}\le \pp {\tilde s} {\reset\sigma,\tilde\tau} {T=\infty}+\delta. 
\end{align}

Let $\tilde n\in\nat$ the number from the statement of \Cref{lem:general <1}, thus
\begin{align}
  \label{eq:def d}
  d\defeq \pp {\tilde s} {\reset\sigma,\tilde\tau} {F_{\tilde n} <\infty} < 1. 
\end{align}
And from \eqref{eq:sup}, some basic properties of expectations we deduce:
\begin{align*}
  \pp {\tilde s} {\reset \sigma,\tilde\tau} {T=\infty} & = \ee {\tilde s} {\reset\sigma,\tilde \tau} {\pp {\tilde s }{\reset\sigma,\tilde\tau}{T=\infty\ \vert\ F_{\tilde n},S_0,\ldots,S_{F_{\tilde n}}}}\\
                                                       & = \ee {\tilde s} {\reset\sigma,\tilde \tau} {\mathds{1}_{F_{\tilde n}<\infty}\cdot \pp {\tilde s }{\reset\sigma,\tilde\tau}{T=\infty\ \vert\ F_{\tilde n},S_0,\ldots,S_{F_{\tilde n}}}}\\
                                                       & = \ee {\tilde s} {\reset\sigma,\tilde \tau} {\mathds{1}_{F_{\tilde n}<\infty}\cdot \pp {S_{F_{\tilde n}}}{\reset\sigma,\tilde\tau[S_0\cdots S_{F_{\tilde n}}]}{T=\infty}}\\
                                                       & \le \ee {\tilde s} {\reset\sigma,\tilde \tau} {\mathds{1}_{F_{\tilde n}<\infty}\cdot \left(\pp {\tilde s}{\reset\tau,\tilde \tau}{T=\infty}+\delta\right)}\\
  & = d\cdot \left(\pp {\tilde s}{\reset\sigma,\tilde \tau}{T=\infty}+\delta\right).
\end{align*}
The random variable $S_{F_{\tilde n}}$ is well-defined because we are measuring the infinite plays where $F_{\tilde n}$ is finite; on the third equality we have used the definition of the reset strategy and the last two (in)equalities we have used \eqref{eq:sup} and \eqref{eq:def d} respectively. Since $d<1$ then we have
\begin{align*}
  \pp {\tilde s}{\reset\sigma, \tilde \tau}{T=\infty} \le \frac d {1-d}\delta,
\end{align*}
so for all states $s'$ and strategies $\tau'$ it follows that
\begin{align*}
  \pp {s'} {\reset\sigma,\tau'} {T=\infty} \le  \sup_{\tau,s} \pp s {\reset\sigma,\tau} {T=\infty} 
                                           \le \pp {\tilde s} {\reset\sigma, \tilde \tau} {T=\infty} +\delta \le \frac \delta {1-d}. 
\end{align*}
Since this holds for any $\delta>0$, \eqref{eq:third property} follows. 

\subsubsection{$\epsilon$-Optimal}
\label{sub:eps opt}
The last property of $\sigma$ that we have to prove is that if we assume that it has the previous properties, namely that it is $\epsilon$-optimal, locally optimal, and it has finitely many drops, then the reset strategy $\reset\sigma$ is $\epsilon$-optimal as well.
So fix an $\epsilon>0$ and a strategy $\sigma$ that is both locally optimal and $\epsilon$-optimal, and for which \eqref{eq:third property} holds. We define for all naturals $n$, strategies $\reset \sigma_n$ that reset only up to date $n$, and prove that they are $\epsilon$-optimal first. 

Define $\mathfrak T_n$ to be the function that truncates finite plays to length $n$:
\begin{align*}
  \mathfrak T_n (s_0\cdots s_m) \defeq \begin{cases}
    s_0\cdots s_m\ &\text{ if }m\le n,\\
    s_0\cdots s_n\ &\text{ otherwise.}
  \end{cases}
\end{align*}
The reset strategy that resets only up to date $n$ is then defined as:
\begin{align*}
  \reset\sigma_n(s_0\cdots s_m)\defeq\sigma(s_\ell\cdots s_m),
\end{align*}
where
\begin{align*}
  \ell\defeq \Lambda(\epsilon,\sigma)\left(\mathfrak T_n(s_0\cdots s_m)\right). 
\end{align*}
\begin{lemma}
  \label{lem:sigma n}
  For all $n\in\nat$, $\reset\sigma_n$ is $\epsilon$-optimal.
\end{lemma}
\begin{proof}
  The proof is by induction on $n$. The base case is trivial since $\reset\sigma_0=\sigma$, therefore assume that the lemma is true for $n-1$, we prove that it is also true for $n$. Namely we fix a state $s$ and a strategy $\tau$ and prove that
  \begin{align*}
    \ee s {\reset\sigma_n,\tau} f \ge \val(s)-\epsilon. 
  \end{align*}
  Denote by $\mathcal E$ the event that there is a drop at date $n$, and by $\mathfrak D$ the set of finite plays of length $n$ that are drops. Then we have
  \begin{align*}
    \ee s {\reset\sigma_n,\tau} f &= \ee s {\reset\sigma_n,\tau} {\mathds{1}_{\mathcal E}\cdot f} + \ee s {\reset\sigma_n,\tau} {\mathds{1}_{\neg\mathcal E}\cdot f}\\
                                  &= \sum_{s_0\cdots s_n\in \mathfrak D}\pp s {\reset\sigma_n,\tau} {s_0\cdots s_n}\cdot \ee {s_n} {\sigma,\tau[s_0\cdots s_n]} f + \ee s {\reset\sigma_n,\tau} {\mathds{1}_{\neg\mathcal E}\cdot f}\\
    &\ge \sum_{s_0\cdots s_n\in \mathfrak D}\pp s {\reset\sigma_n,\tau} {s_0\cdots s_n}\cdot \left(\val(s_n)-\epsilon\right)+ \ee s {\reset\sigma_n,\tau} {\mathds{1}_{\neg\mathcal E}\cdot f},
  \end{align*}
  where in the second equality we have used the definition of $\reset\sigma_n$ and in the inequality the $\epsilon$-optimality of $\sigma$. Define the strategy $\tau'$ to be the strategy that plays like $\tau$ except if there is a drop at date $n$, in which case it resets to a $\epsilon/2$-response called $\tau''$. Then we have
  \begin{align*}
    \ee s {\reset\sigma_{n-1},\tau'} f &= \ee s {\reset\sigma_{n-1},\tau'} {\mathds{1}_{\mathcal E}\cdot f}+\ee s {\reset\sigma_{n-1},\tau'} {\mathds{1}_{\neg\mathcal E}\cdot f}\\
                                    &= \sum_{s_0\cdots s_n\in\mathfrak D}\pp s {\reset\sigma_{n-1},\tau'} {s_0\cdots s_n}\cdot \ee {s_n} {\reset\sigma_{n-1}[s_0\cdots s_n],\tau''} f +\ee s {\reset\sigma_{n-1},\tau} {\mathds{1}_{\neg\mathcal E}\cdot f}\\
    &\le \sum_{s_0\cdots s_n\in\mathfrak D}\pp s {\reset\sigma_{n-1},\tau'} {s_0\cdots s_n}\cdot \left(\val(s_n)-2\epsilon+\epsilon/2\right)+\ee s {\reset\sigma_{n-1},\tau} {\mathds{1}_{\neg\mathcal E}\cdot f}.
  \end{align*}
  Now since the strategies $\reset\sigma_{n-1}$ and $\reset\sigma_n$ on one hand, and strategies $\tau$ and $\tau'$ on the other, behave the same for all plays of length smaller than $n$ and on infinite plays where there is no drop at date $n$, it follows that the right-most terms in the two inequalities above, as well as the factors $\ppdd$ on the left are equal. Consequently we can combine the two inequalities above to conclude that
  \begin{align*}
    \ee s {\reset\sigma_{n-1},\tau'} f \le \ee s {\reset \sigma_n,\tau} f. 
  \end{align*}
  This concludes the induction step and the proof of the lemma.
\end{proof}

We now prove that
\begin{align}
  \label{eq:reset epsilon optimal}
  \text{$\reset\sigma$ is $\epsilon$-optimal},
\end{align}
the final property of $\sigma$ given in the beginning of this section.

Let $m$ respectively $M$ be a lower bound, respectively upper bound of the payoff function. Define $T$ to be the stopping time that is equal to the date of the last drop if it exists otherwise it is equal to $\infty$. 

Applying \Cref{lem:sigma n} we have that for all $n\in\nat$, $s$, and $\tau$
\begin{align*}
  \val(s)-\epsilon \le \ee s {\reset\sigma_n,\tau} {\mathds{1}_{T\le n}\cdot f} + \ee s {\reset\sigma_n,\tau} {\mathds{1}_{T>n}\cdot f}
  \le \ee s {\reset\sigma_n,\tau} {\mathds{1}_{T\le n}\cdot f} + M\cdot \pp s {\reset\sigma_n,\tau} {T>n}.
\end{align*}
Since $\reset\sigma$ and $\reset\sigma_n$ behave the same on the plays in the event $T\le n$, we have that for all $n\in\nat$, $\tau$ and $s$
\begin{align*}
  \ee s {\reset\sigma,\tau} f - \ee s {\reset\sigma,\tau} {\mathds{1}_{T>n}\cdot f} &= \ee s {\reset\sigma,\tau}{\mathds{1}_{T\le n}\cdot f}= \ee s {\reset\sigma_n,\tau}{\mathds{1}_{T\le n}\cdot f}\\
  &\ge \val(s)-\epsilon-M\cdot \pp s {\reset\sigma_n,\tau} {T>n}.
\end{align*}
The strategies $\reset\sigma_n$ and $\reset\sigma$ behave the same on plays in the event $T\le n$, and therefore also on those in the event $T>n$; consequently we can write
\begin{align*}
  \pp s {\reset\sigma_n,\tau} {T>n} = \pp s {\reset\sigma,\tau} {T>n}, 
\end{align*}
and from the inequality above we have:
\begin{align*}
  \ee s {\reset\sigma,\tau} f \ge \val(s)-\epsilon - (M-m)\cdot \pp s {\reset\sigma,\tau} {T>n}.
\end{align*}
From the sigma-additivity of measures and the property in \eqref{eq:third property} it follows that
\begin{align*}
  \lim_{n\to\infty}(M-m)\cdot \pp s {\reset\sigma,\tau} {T>n} = 0. 
\end{align*}
Since $\tau$ and $s$ are general, this proves $\epsilon$-optimality of $\sigma$, that is it proves the final property \eqref{eq:reset epsilon optimal}. \Cref{lem:2e subgame} in conjunction with \eqref{eq:reset epsilon optimal} implies \Cref{thm:subgameperfect}.

\subsubsection{Remark on Optimal Strategies}

Martin's theorem, \Cref{thm:martin} implies that the games that we are interested in have $\epsilon$-optimal strategies for every $\epsilon>0$. We have then proved that there are locally optimal (2) strategies that are also $\epsilon$-optimal (1). We then showed that for strategies with properties (1) and (2), we can prove that they also posses the properties (3) and (4), which respectively stated that there are finitely many drops and that the reset strategy is also $\epsilon$-optimal. By inspection, in the proofs of 
\begin{align*}
  (1)\text{ and }(2)\qquad &\Rightarrow \qquad (3),\\
  (1), (2)\text{ and } (3)\qquad &\Rightarrow \qquad (4), 
\end{align*}
in \Cref{sub:finitely many drops} and \Cref{sub:eps opt} respectively, the variable $\epsilon$ need not be strictly positive. Since optimal strategies are necessarily locally optimal, the following lemma follows from \Cref{lem:2e subgame}.

%\begin{remark}
%  \label{rem:optimal strategies}
%  If $\sigma$ is an optimal strategy then $\reset\sigma$ is subgame-perfect. 
%\end{remark}

We can summarize our results for $\epsilon=0$ or $\epsilon>0$ as:
\begin{lemma}\label{lem:subgame}
Let $\game$ be a game equipped with a shift-invariant payoff function.
Let $\epsilon \geq 0$ be a non-negative real number
and $\sigma$ be an $\epsilon$-optimal strategy in $\game$.
Assume that $\sigma$ is locally optimal.
Then the reset strategy $\hat\sigma$ is $2\epsilon$-subgame-perfect in $\game$. 
\end{lemma}

\section{Half-Positional Games}
\label{sec:main theorem}

We prove the main theorem:
\main*
Neither of the conditions in the statement is necessary, as we saw from the examples given in~\Cref{sec:applications}. Necessary and sufficient conditions for positionality are known for deterministic games~\cite{concur}. However the shift-invariant and submixing conditions are general enough to recover several known classical results, and to provide several new examples of games with deterministic stationary optimal strategies. Before we proceed with the proof we remark:
\begin{remark}
  \label{rem:inverse}
  A symmetric proof to that of \Cref{thm:main}, the subject of this section, can be used to prove a statement like that of \Cref{thm:main}, where Player~1 is replaced by Player~2 and submixing is replaced by inverse-submixing. A corollary of this is that games with shift-invariant, submixing and inverse-submixing payoff functions are positional. 
\end{remark}

Consider a game $\game$ fulfilling the conditions of the theorem. The proof proceeds by induction on the actions of the maximizer, that is on the quantity
\begin{align*}
  N(\game)\defeq \sum_{s\in\states_1}\left(|\actions(s)|-1\right).
\end{align*}
It proceeds by removing more and more actions of the maximizer and showing that at every step the value has not decreased, until we are left with a single choice from every state that belongs to the maximizer. The unique choice will then be the positional optimal strategy. 

If $N(\game)=0$ there is no choice for maximizer, hence he has a deterministic and stationary optimal strategy. 
If $N(\game)>0$ there must be a state $\tilde s\in\states$ such that Player~1 has at least two actions in $\tilde s$, \ie $\actions(\tilde s)$ has at least two elements. We split the game $\game$ in two strictly smaller subgames $\game_1$ and $\game_2$.

\begin{definition}[Split of a game]
Let $\game$ be a game with $N(\game)>0$ and $\tilde s\in\states$ a state of $\game$
controller by Player~1 in which there are at least two actions available,
\ie $\actions(\tilde s)$ has at least two elements. 
 Partition $\actions(\tilde s)$ into two non-empty sets: $\actions_1$ and $\actions_2$.  Let $\game_1$ and $\game_2$ be the games obtained from $\game$ by restricting the actions in the state $\tilde s$ to $\actions_1$ and $\actions_2$ respectively. 
 Then $(\game_1,\game_2)$ is called a \emph{split} of $\game$ on $\tilde s$.
 \end{definition}

The induction step relies on the two results stated in the next theorem.
The first result says that the value of $\tilde s$ in the original game cannot be larger than that of the restricted games.
The second result shows that Player $1$ can play optimally in $\game$
by selecting one of the subgames and play optimally in it.
\begin{theorem}\label{thm:fundamentalineq}
Let $\game$ be a game equipped with a payoff function that is shift-invariant and submixing.
Let $(\game_1,\game_2)$ a \emph{split} of $\game$ on $\tilde s$.
Then
\begin{align}
  \label{eq:max}
  \val(\game)(\tilde s) = \max\set{\val(\game_1)(\tilde s),\val(\game_2)(\tilde s)}. 
\end{align}
Assume moreover that $\val(\game_1)(\tilde s)\geq \val(\game_2)(\tilde s)$.
Then, for every $s\in\states$,
\begin{align}
  \label{eq:last goal}
  \val(\game)(s)=\val(\game_1)(s). 
\end{align}
\end{theorem}

Theorem~\ref{thm:main} is a simple corollary of Theorem~\ref{thm:fundamentalineq}.
\begin{proof}[Proof of \Cref{thm:main}]
The proof is by induction on $N(G)$.
If $N(\game)=0$ there is no choice for maximizer, hence he has a deterministic and stationary optimal strategy.
If $N(\game)>0$ then we choose a split $(\game_1,\game_2)$ of $\game$
on a pivot state $\tilde s$.
By symmetry, we can choose a split such that $\val(\game_1)(\tilde s)\geq \val(\game_2)(\tilde s)$.
Then, according to~\eqref{eq:last goal} in Theorem~\ref{thm:fundamentalineq}, a strategy for Player $1$
which is optimal in $\game_1$ is also optimal in $\game$.
By induction hypothesis,
there exists a positional optimal strategy in $\game_1$,
thus $\game$ is half-positional.
\end{proof}

The rest of the section is dedicated to the proof of Theorem~\ref{thm:fundamentalineq}.
We fix a game $\game$ 
and a split $(\game_1,\game_2)$ of $\game$ on the state $\tilde s$.
The inequality 
\[
\val(\game)(\tilde s) \geq \max\set{\val(\game_1)(\tilde s),\val(\game_2)(\tilde s)}
\]
is clear, since Player~1 has more choice in $\game$ than he has in $\game_1$ and $\game_2$.
%We fix for the rest of the section a game $\game$ and a split $(\game_1,\game_2)$ a \emph{split} of $\game$ on $\tilde s$.
We witness the converse inequality with a strategy for Player~2, called the \emph{merge} strategy,
which merges two $\epsilon$-subgame-perfect strategies in the respective smaller games.
This is done in Section~\ref{subsec:merge}. 
The definition of the merge strategy hinges on the projection of plays in the main game $\game$
to plays in the restricted games $\game_1$ and $\game_2$, which is done in section~\ref{subsec:proj}.
Then we analyse  the two possible outcomes: (a) after some date the play remains only in game $\game_1$ (or only in game $\game_2$), (b) the play switches infinitely often between the two smaller games.
This analysis is performed in sections~\ref{subsec:stay} and ~\ref{subsec:switch}.
For the latter case (b) we use the submixing property to show 
that Player $1$ cannot get a better payoff by switching between the two smaller games
that he could get by staying in one of the subgames.

\subsection{Projecting a play in $\game$ to a couple of plays in the subgames}
\label{subsec:proj}

There is a natural way
to project a play $h$ of the game $\game$ starting in $\tilde s$ to a couple of plays
$h_1$ and $h_2$ in the restricted games $\game_1$ and $\game_2$ respectively,
starting from $\tilde s$ as well.
The two projections are computed simultaneously and inductively.
Initially, $h=\tilde s$ and both projections $h_1$ and $h_2$ are also equal to $\tilde s$.
Each step of the play in $\game$
is appended to either $h_1$ or $h_2$, depending on the action $a$ played 
the last time the state $\tilde s$ was visited: if $a$ belongs to $\actions_1$
then the new step is appended to $h_1$, otherwise it is appended to $h_2$.
The computation of $h_1$ and $h_2$ is illustrated on Figure~\ref{fig:proj}.

\begin{figure}
\begin{center}
\includegraphics[width=.9\textwidth]{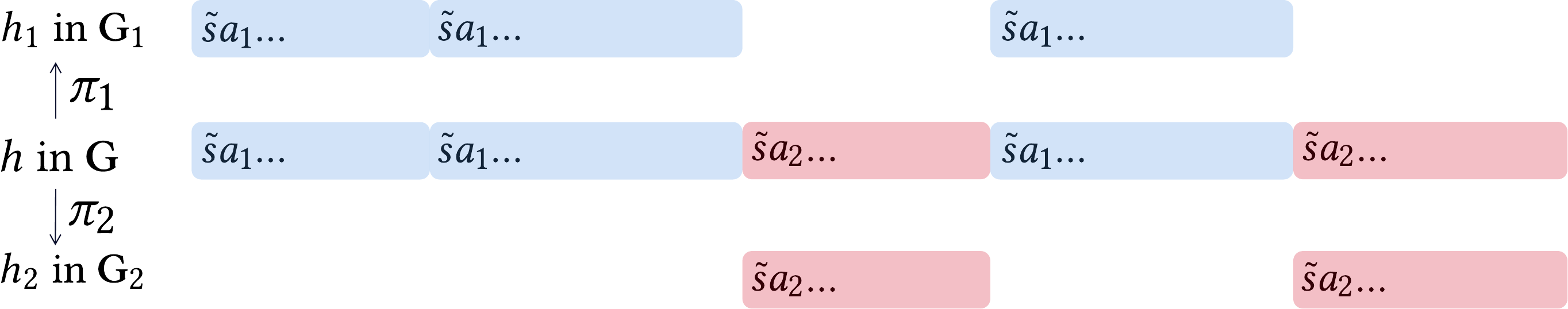}
\end{center}
\caption{\label{fig:proj}The play $h$ is the concatenation of finite plays starting in $\tilde s$,
represented by blocks whose colours depend on the first action played after $\tilde s$,
blue if the action belongs to $\actions_1$ and pink if it belongs to $\actions_2$.
The projection $h_1=\pi_1(h)$ in $\game_1$ is the concatenation of the blue blocks while $h_2=\pi_2(h)$
is the concatenation of the pink blocks.
The projections lose some information about the play in $\game$: swapping
two contigous blocks of different colors in $h$ does not modify
the projections $h_1$ and $h_2$.
}
\end{figure}

Formally, we define two maps $\pi_1$, $\pi_2$ from finite plays in $\game$
starting from $\tilde s$ 
to finite plays in $\game_1$ and $\game_2$ respectively,
starting from $\tilde s$ as well.
Let $h = s_0a_0s_1\ldots s_n$  be a finite play in $\game$ starting in $\tilde s$
and $ha s$ a continuation of $h$ in $\game$, with one more transition $(s_n,a,s)$.
Let $\last(has)$ be the action played in $has$ after the last visit to $\tilde s$ i.e.
\[
\last(has) = a_{\max \{ j \in 0\ldots n \ \mid \ s_j = \tilde s  \}}
=
\begin{cases}
a & \text{ if $s_n=\tilde s$}\\
\last(h) & \text{ otherwise.}
\end{cases}
\]
Then
\[
\pi_1(has)
=
\begin{cases}
\pi_1(h)a s & \text{if $\last(has ) \in \actions_1$}\\
\pi_1(h) & \text{if $\last(has ) \in \actions_2$}\enspace.
\end{cases}
\]
And $\pi_2$ is defined symmetrically with respect to $\actions_1$ and $\actions_2$.

This definition can be extended to infinite plays in a natural way.
Let $h=s_0a_0s_1\ldots$ be an infinite play in $\game$ starting in $\tilde s$.
Then $\pi_1(h)$ is the limit of the sequence
\begin{align*}
 \left(\pi_1(s_0a_0s_1\ldots s_n)\right)_{n\in\nat}\enspace.
\end{align*}
The projection $\pi_1(h)$ can be either finite or infinite,
depending whether the play ultimately stays in $\game_2$ or not.
If after some time the last action chosen in $\tilde s$ is always in $\actions_2$,
all subsequent moves in $\game$ are appended to the projection in $\game_2$,
while the projection to $\game_1$ never gets updated and stays finite.

\subsection{Linking the payoff in $\game$ to the payoffs in the subgames}
The  payoff in $\game$ can be related to the payoff in the subgames. 
We introduce the events
\begin{align*}
\stayntwo&\defeq\{
\forall m\geq n, \last(S_0A_0\ldots S_mA_{m}S_{m+1})\in\actions_2
\}\\
\staytwo&\defeq \bigcup_{n\in \nat} \stayntwo
\enspace.
\end{align*}
If $\stayntwo$ holds, we say that the play \emph{stays in $\game_2$ after step $n$}
whereas if $\staytwo$ holds, we say that the play \emph{ultimately stays in $\game_2$}.

Those two events can be described equivalently as a non-update
of the projection to $\game_1$ after some point.
For that, we make use of the random variables:
\begin{align*}
 \Pi \defeq S_0A_0S_1\cdots
 \qquad
 \Pi_1 \defeq \pi_1(S_0A_0S_1\cdots),\qquad \Pi_2 \defeq \pi_2(S_0A_0S_1\cdots).
\end{align*}
Recall that $S_n$ and $A_n$ are the random variables
which output respectively the $n$-th state $s_n$ and action $a_n$ when the play is
$s_0a_0s_1a_1\cdots$.
%In other words $S_n$ and $A_n$ are respectively the $n$th state and action in the trajectory. 
We see that $\Pi$ is simply the identity map outputing the play in $\game$
while
%$\Pi$ is a random variable taking values in $\states(\actions\states)^\omega$, that is given by the identity map, and 
$\Pi_i$ is essentially equivalent to $\pi_i$, it is a random variable that maps the infinite play
in game $\game$ to its finite or infinite projection in game $\game_i$.
Then
%$\stayntwo$ holds if $\Pi_1 = \pi_1(S_0A_1\cdots S_n)$.
\begin{align*}
&\stayntwo
=\{\Pi_1 = \pi_1(S_0A_1\cdots S_n)\}\\
&
\staytwo
=\{\Pi_1 \text{ is finite}\}
\enspace.
\end{align*}

The events $\staynone$ and $\stayone$ are defined symmetrically.
Define the event
\[
\switch \defeq \left(\neg \stayone \land \neg \staytwo\right) = \{ \text{ both $\Pi_1$ and $\Pi_2$ are infinite } \}\enspace.
\]
The following lemma shows that the payoff in $\game$ is tightly related to the payoffs in the subgames
$\game_1$ and $\game_2$.
\begin{lemma}\label{rem:proj}
Let $f$ be a prefix-independent and submixing payoff function.
Every infinite play in $\game$ belongs to exactly one of the three events $\{\stayone,\staytwo,\switch\}$.
Moreover,
\begin{align}
\label{eq:finiterem}
&\text{if } \stayone \text{ holds then } f(\Pi) = f(\Pi_1)\enspace.\\
\label{eq:finiterem2}
&\text{If } \staytwo \text{ holds then }   f(\Pi) = f(\Pi_2)\enspace.\\
\label{eq:infiniterem2}
&\text{If } \switch \text{ holds then }  f(\Pi) \leq \max(\ f(\Pi_1)\ ,\ f(\Pi_2)\ )
\enspace.
\end{align}
\end{lemma}
\begin{proof}
Since both projections in $\game_1$ and $\game_2$ cannot be finite at the same time
then
$(\stayone,\staytwo,\switch)$ is a partition of the infinite plays in $\game$.
If $\Pi_1$ is finite then $\Pi$ and $\Pi_2$ share an infinite suffix and the prefix-independence of $f$ implies~\eqref{eq:finiterem}.
The case where $\Pi_2$ is finite is symmetric, hence~\eqref{eq:finiterem2}.
If both $\Pi_1$ and $\Pi_2$ are infinite then the sequence of actions
$(\last(S_0\ldots S_nA_{n}S_{n+1}))_{n\in\nat}$
switches infinitely often between $\actions_1$ and $\actions_2$ thus $\tilde s$ 
is visited infinitely often.
Moreover, in this case $\Pi$ is a shuffle of $\Pi_1$ and $\Pi_2$
and since $f$ is submixing,~\eqref{eq:infiniterem2} follows.
\end{proof}

\subsection{The Merge Strategy}
\label{subsec:merge}

In light of Lemma~\ref{rem:proj}, it is intuitively clear that to play well in $\game$,
Player $2$ has to play well in both subgames $\game_1$ and $\game_2$.
Fix $\epsilon>0$.
The \emph{merge strategy}
for Player $2$ is the composition of two 
strategies $\topt_1$ and $\topt_2$
for Player~2 in the subgames
$\game_1$ and $\game_2$ respectively.
We require $\topt_1$ and $\topt_2$ to be
$\epsilon$-subgame-perfect in the corresponding subgames;
%and value-preserving
their existence is guaranteed by \Cref{thm:subgameperfect}.

\begin{definition}
 The merge strategy $\topt$ is a strategy in $\game$ for Player $2$
 which ensures
that $\Pi_1$ is consistent with $\topt_1$
and
 $\Pi_2$ is consistent with $\topt_2$ when the play starts from $\tilde s$.
Let $h$ be a finite play in $\game$ from $\tilde s$ and ending in a state controlled by Player $2$,
then
\begin{align*}
\topt(h) = \begin{cases}
  \topt_1(\pi_1(h)) & \text{ if\ \   $\last(h)\in \actions_1$\enspace,}\\
  \topt_2(\pi_2(h)) & \text{ if\ \  $\last(h)\in \actions_2$}\enspace.
\end{cases}
\end{align*}
\end{definition}
The merge strategy is well-defined because if $\last(h)\in\actions_1$ then both $h$ and $\pi_1(h)$ end with the same state,
controlled by Player $2$.
%\subsection{Proof of Theorem~\ref{thm:fundamentalineq}}

In the next two sections,
we show that the merge strategy
guarantees to Player $2$ some upper-bounds on the expected payoffs,
which reflect the bounds given in Lemma~\ref{rem:proj} for payoffs of individual plays.

\subsection{On plays consistent with the merge strategy and ultimately staying in $\game_2$}
\label{subsec:stay}
%For any pure strategy $\sigma$, 

%Next lemma control the expected payoff in case the projection $\Pi_1$ is finite,
%i.e. if after some visit to $\tilde s$,
%the rest of the play stays in the game $\game_2$.
%Consider the event
%%Say that an infinite play 
%%$\mathfrak C^1$ to be the set of finite plays $h$ in $\game_2$ ending in the state $\tilde s$, and such that $\sigma(h)$ wants to play an action in $\actions_1$; \ie prefixes where a switch to $\game_1$ is about to occur. The set of infinite plays without any prefix in $\mathfrak C^1_\sigma$ is denoted
%%The event that is generated by the cylinders of prefixes in $\mathfrak C^1_\sigma$ (respectively its complement) abbreviates as
%\begin{align*}
%  \text{stay in $\game_2$}\enspace,
%%  \text{leaves $\game_2$},
%  %\qquad ( \text{respectively stay in $\game_2$})
%\end{align*}
%which requires that every prefix of the play ending in $\tilde s$
%is followed by an action in $\actions_2$,
%or equivalently asking that $\Pi_1$ is empty.
%%i.e. if stay in $\game_2$.
%%%One can define symmetrically the events for $\game_1$.

In this section,
we show that in case the play ultimately stays in $\game_2$,
then the expected payoff is upper-bounded by
$\val(\game_2)(\tilde s) + \epsilon$.

For simplicity, we require $\epsilon$ to be small enough so that
$\tau_2$ does not select any \emph{value-increasing action},
in the following sense.
%because those actions are a strategic gift to Player $1$.
\begin{lemma}
In $\game_2$, fix a state $s$ controlled by Player $2$
and an action $a$  available in that state.
Denote 
\[
\delta(s,a) = \left(\sum_{t\in S} p(s,a,t) \val(\game_2)(t)\right) - \val(\game_2)(s)\enspace.
\]
Then $\delta(s,a)\geq 0$.

In case $\delta(s,a) > 0$ then
$a$ is said to be \emph{value-increasing} in $s$.
In that case, if moreover $\epsilon$ is strictly smaller that $\delta(s,a)$,
then $\topt_2$ never selects the action $a$ in a play ending in state $s$.
\end{lemma}
\begin{proof}
Since the payoff function is prefix-independent,
and $s$ is controlled by Player $2$,
then $\delta(s,a)\geq 0$, because after Player $2$
chooses $a$ in $s$, he can proceed with 
an
$\epsilon'$-optimal strategy from the states $t$ such that $p(s,a,t)>0$,
for an arbitrary $\epsilon'>0$.
Assume $\epsilon$  strictly smaller that $\delta(s,a)$.
Then $\topt_2$ 
never selects $a$ in state $s$, otherwise this would contradict
the $\epsilon$-subgame perfection of $\tau_2$:
Player $1$ could proceed with some
$(\delta(s,a) - \epsilon)/2$-optimal strategy in $\game_2$
and get an expected payoff strictly greater than $\val(\game_2)(s)+\epsilon$.
\end{proof}

\begin{lemma}
  \label{lem:fin}
Assume that $f$ is prefix-independent and
 $\epsilon$ is small enough
to guarantee that $\topt_2$ never selects any value-increasing action.
Let $\sigma$ be a strategy for Player $1$ in $\game$
such that
$\pp{\tilde s}{\sigma,\topt}{  \staytwo} > 0$.
Then
  \begin{align}
    \label{eq:sfin}
    \ee{\tilde s}{\sigma,\topt}{f  \mid   \staytwo} 
    \leq \val(\game_2)(\tilde s) + \epsilon
    \enspace.
  \end{align}
\end{lemma}

\begin{proof}
%We rule out the case where $\pp{\tilde s}{\sigma,\topt}{\staytwo}=0$,
%since in that case~\eqref{eq:sfin} is undefined.
The first ingredient of the proof is the sequence
of random variables $(V_n)_{n\in \nat}$,
where $V_n$ denotes the value in $\game_2$ of the last vertex of $\pi_2(S_0A_1\cdots S_n)$. 
Since the play starts in state $\tilde s$, 
\[
V_0 = \val(\game_2)(\tilde s)\enspace.
\] 
The value of $V_n$ does not change unless the projection
of the play to $\game_2$ via $\pi_2$ does.
Since $\Pi_2$ is consistent with $\topt_2$
and since 
$\topt_2$ never selects any value-increasing action,
%Then, if Player $1$ chooses to play an action in $A_1$,
%the value of $V_n$ does not change till the play reaches $\tilde s$ again,
%and Player $1$ chooses to play an action in $A_2$.
%%and $v^2_n$ stays equal to $v^2_0$. Anytime the play reaches $\tilde s$ and is followed by 
%Then, 
\[
(V_n)_{n\in \nat}\text{ is a super martingale}\enspace.
\]

%forbids $\topt_2$ to choose actions which are not value preserving.

%Denote $f^1_n$ the value in $\game_1$ of the last vertex of $\pi_1(H_n)$.
%Since both $\top_0$ and $\top_1$ are value-reserving,
%then both $(f^0_n)_n$ and $(f^1_n)_n$ are bounded super-martingales.

The second ingredient in the proof is a stopping time $T$,
defined as follows.
%Remark that $\stayntwog0\subseteq \stayntwog1 \subseteq$ and
%\begin{align}
%\label{eq:En}
%\{\Pi_1\text{ is finite }\} = \bigcup_{n\in\nat} \stayntwo\enspace.
%\end{align}
For every finite play $h=s_0a_0\ldots s_n$  in $\game$
starting in $\tilde s$ and consistent with $\sigma$ and $\topt$,
denote 
\[
\phi(h) = \pp{\tilde s}{\sigma,\topt}{\stayntwo \mid h\text{ is a prefix of the play}}\enspace.
\]
Fix some $\epsilon'>0$ and denote
%\[
%H = \{\text{ finite plays $h$ from $\tilde s$ } 
%\mid \phi(h) \geq 1 - \epsilon'\}\enspace.
%\]
$T$ the stopping time 
%the smallest $n$
%such that
\[
T = \min \left\{ n \in \nat \mid  \phi(S_0A_0\ldots S_n) \geq 1 - \epsilon' \right\}\enspace,
\]
with the usual convention $\min(\emptyset) = \infty$.

\medskip

We use the event $\{T < \infty \}$ as an approximation of the event $\staytwo$ by proving
  \begin{align}
  \label{eq:approx2}
&    \pp{\tilde s}{\sigma,\topt}{ \staytwo \mid T < \infty} \geq 1 - \epsilon'    \\
  \label{eq:approx1}
&    \pp{\tilde s}{\sigma,\topt}{T < \infty\mid \staytwo} = 1
    \enspace.
  \end{align}
  The inequality~\eqref{eq:approx2} holds because 
  by definition of $\phi$, for every $n\in\nat$,
  \[
  \pp{\tilde s}{\sigma,\topt}{ \stayntwo \mid T = n} \geq 1 - \epsilon' \enspace.
  \]
%  
%  
%   is straightforward
%  from the definition of $\phi$ and the fact that 
%  \begin{align}
%  \label{eq:En}
%  \bigcup_n \stayntwo = \{ \Pi_1 \text{ is finite}\}\enspace.
%  \end{align}
%Let $h$ be a finite play such that $\phi(h)\geq 1 - \epsilon'$.
%  Then
%  $\pp{\tilde s}{\sigma,\topt}{ \Pi_1 \text{ is finite} \mid h} \geq 1 - \epsilon'$.
We show~\eqref{eq:approx1}.
Fix $\epsilon''>0$.
By definition of $\staytwo$,
there exists $n_1\in\nat$ such that
\begin{align}
\label{eq:epsilonpp}
\pp{\tilde s}{\sigma,\topt}{
\stayntwog{n_1}\mid \staytwo 
  }\geq 1 - \epsilon''\enspace.
 \end{align}
According to L\'evy's 0-1 law (see \eg \cite[Theorem 14.4]{wiliamspr}), the sequence of random variables
$\left(\ee{\tilde s}{\sigma,\topt}{\stayntwog{n_1} \mid S_0,\ldots , S_n}\right)_{n\in\nat}$
almost-surely converges to the indicator function 
${\bf 1}_{\stayntwog{n_1}}$.
Thus, 
\[\pp{\tilde s}{\sigma,\topt}{
\exists n_2 \geq n_1,
%\in\nat,\forall m \geq n_2,
\ee{\tilde s}{\sigma,\topt}{\stayntwog{n_1} \mid S_0,\ldots , S_{n_2}}
\geq 1 - \epsilon'
\ \vert \ 
\stayntwog{n_1}
}=1\enspace.
\]
Since $n_2\geq {n_1}$ implies $\stayntwog{n_2}\subseteq \stayntwog{n_1}$,
\[\pp{\tilde s}{\sigma,\topt}{
\exists n_2 ,
%\in\nat,\forall m \geq n_2,
\phi(S_0,\ldots,S_{n_2})
\geq 1 - \epsilon'\mid
\stayntwog{n_1}
}=1\enspace.
\]
Equivalently,
\[
\pp{\tilde s}{\sigma,\topt}{T < \infty\mid \stayntwog{n_1}}=1\enspace. 
\]
and with~\eqref{eq:epsilonpp} we get
\[
\pp{\tilde s}{\sigma,\topt}{T < \infty\mid \staytwo}\geq 1 - \epsilon''\enspace. 
\]
This holds for every $\epsilon''>0$, hence~\eqref{eq:approx1}.

\medskip 
Since $\epsilon'>0$ can be chosen arbitrarily small,
then according to~\eqref{eq:approx2} and~\eqref{eq:approx1},
to show our goal~\eqref{eq:sfin}, it is enough to establish:
  \begin{align}
    \label{eq:sfinbis}
    \ee{\tilde s}{\sigma,\topt}{f  \mid   T < \infty} 
    \leq \val(\game_2)(\tilde s) + \epsilon + 2\epsilon' \cdot ||f||_\infty 
    \enspace.
  \end{align}
  This is well-defined,
  because~\eqref{eq:approx1} ensures
$\pp{\tilde s}{\sigma,\topt}{T < \infty}\geq \pp{\tilde s}{\sigma,\topt}{\staytwo}>0$,
and $f$ is bounded.

Since $(V_n)_{n\in \nat}$ is a bounded super martingale,
we can deduce from Doob's Forward Convergence Theorem {\cite[Theorem 11.5]{wiliamspr}}
that $(V_n)_{n\in \nat}$ converges almost-surely.
We denote $V_T$ the random variable equal to 
$(\lim_n V_{n})$ if $T=\infty$ and $V_n$ if $T=n$.

We deduce~\eqref{eq:sfinbis}
from the following three inequalities:
\begin{align}
\label{eq:fin1}
&\ee{\tilde s}{\sigma,\topt}{V_T} \leq  \val(\game_2)(\tilde s)
\\
\label{eq:fin3}
&
\ee{\tilde s}{\sigma,\topt}{V_T \mid T = \infty} 
=
\val(\game_2)(\tilde s)
\\
\label{eq:fin2}
&\ee{\tilde s}{\sigma,\topt}{f \mid T < \infty} 
\leq
\ee{\tilde s}{\sigma,\topt}{V_T \mid T < \infty} + \epsilon
+ 2\epsilon' \cdot ||f||_\infty 
\enspace.
\end{align}
Assuming~\eqref{eq:fin1} and~\eqref{eq:fin3} do hold,
then $\ee{\tilde s}{\sigma,\topt}{V_T \mid T < \infty} 
\leq
\val(\game_2)(\tilde s)$. Injecting this inequality in~\eqref{eq:fin2},
we get~\eqref{eq:sfinbis}, and
    the lemma is proved.

\medskip 
We prove the three inequalities~\eqref{eq:fin1}-~\eqref{eq:fin2}.
The inequality~\eqref{eq:fin1}
is obtained using
the equality $V_0 = \val(\game_2)(\tilde s)$
and applying
 Doob's Optional Stopping Theorem
(Section 10.10 in \cite{wiliamspr})
to the bounded super-martingale
$(V_n)_{n\in\nat}$ and the stopping time $T$,
which implies
$\ee{\tilde s}{\sigma,\topt}{V_T} \leq V_0$.

To prove~\eqref{eq:fin3},
we prove an even stronger statement:
\[
\pp{\tilde s}{\sigma,\topt}{V_T = \val(\game_2)(\tilde s) \mid T = \infty} 
= 1\enspace.
\]
If $T=\infty$ then, according
to~\eqref{eq:approx1},
the event $\staytwo$ does not hold.
%play does not ultimately stay in $\game_2$.
Thus,
according to Lemma~\ref{rem:proj},
either $\stayone$ or $\switch$ holds.
%\Pi_2$ is finite or the play $\Pi_2$ visits
%infinitely often $\tilde s$.
In the first case,
$(V_n)_{n\in\nat}$ is ultimately constant equal to $\val(\game_2)(\tilde s)$.
 In the second case, the play $\Pi_2$ visits $\tilde s$ infinitely often.
 Since $(V_n)_n$ converges almost-surely to $V_T$ then $V_T=\val(\game_2)(\tilde s)$.

Finally, we prove~\eqref{eq:fin2}. Denote $h_T$ the random variable
defined when $T$ is finite, which outputs the 
prefix of the play of length $T$, i.e. 
\[
h_T = S_0A_0\ldots S_T\enspace,
\]
and let $h$ such that $\pp{\tilde s}{\sigma,\topt}{h_T=h} > 0 $.
Denote $t$ the last state of $h$.
%the set of infinite plays with a finite prefix $h$ such that $\phi(h) \geq 1 - \epsilon'$.
%Fix such an $h$ and denote $t$ the last state of $h$.
Let $\sigma_0$ be strategy in $\game_2$ which coincides with $\sigma[h]$
as long as the play stays in $\game_2$.
Then:
\begin{align*}
\ee{\tilde s}{\sigma,\topt}{f \mid h_T=h }
&=
\ee{\tilde s}{\sigma,\topt}{f \mid h \text{ is a prefix of the play}}\\
&=
\ee{ t}{\sigma[h],\topt[h]}{
f  }
\\
&\leq
\ee{ t}{\sigma_0,\topt[h]}{
f  }
+ 2\epsilon' \cdot ||f||_\infty
\\
&=
\ee{ t}{\sigma_0,\topt_2[\pi_2(h)]}{
f  }
+ 2\epsilon' \cdot ||f||_\infty
\\
&\leq
\val(\game_2)(t) + \epsilon
+ 2\epsilon' \cdot ||f||_\infty
\\
&=
\ee{\tilde s}{\sigma,\topt}{
V_T\mid h_T=h} + \epsilon
+ 2\epsilon' \cdot ||f||_\infty
\enspace.
\end{align*}
The first equality holds because $\pp{\tilde s}{\sigma,\topt}{h_T=h} > 0 $
thus no strict prefix $h'$ of $h$ satisfies $\phi(h') \geq 1-\epsilon$,
and if $h$ is a prefix of the play then $h_T=h$.
The second equality holds by prefix-independence of $f$.
The first inequality holds because  $\phi(h) \geq 1 - \epsilon'$ thus the strategies $\sigma[h]$ 
and $\sigma_0$ coincide with probability $\geq 1 - \epsilon'$,
and when they do not the payoff difference is at most $2||f|||_\infty$.
The third equality holds because $\tau[h]$ coincides with $\topt_2[\pi_2(h)]$ when
the play stays in $\game_2$.
The second inequality is by $\epsilon$-subgame optimality of $\topt_2$ in $\game_2$.
The last equality holds by definition of $V_T$.

Since this holds for every possible value $h$ of $h_T$ when $T<\infty$,
and there are at most countably many such values,
the inequality~\eqref{eq:fin2} follows.
\end{proof}

\subsection{On plays consistent with the merge strategy and switching infinitely often between the two subgames}
\label{subsec:switch}
In this section,
we provide an upper-bound on the payoff of plays
which switch infinitely often between $\game_1$ 
and $\game_2$.

\begin{lemma}\label{lem:key}
  \label{lem:sinfty}
  Assume that $f$ is prefix-independent and submixing.
  For all strategies $\sigma$,
\begin{align}
  \label{eq:sinfty}
  \pp{\tilde s}{\sigma,\topt}{f \le \max\set{\val(\game_1)(\tilde s),\val(\game_1)(\tilde s)} + \epsilon\ \vert\ 
  \switch}=1.
\end{align}
\end{lemma}

\begin{proof}
By definition of $\switch$,
if $\switch$ occurs
then both projections $\Pi_1$ and $\Pi_2$ are infinite
and visit $\tilde s$ infinitely often.
According to the inequality~\eqref{eq:infiniterem2} in \Cref{rem:proj},
to prove~\eqref{eq:sinfty} it is enough to show,
for every $i\in \{1,2\}$,
  \begin{align}
    \label{eq:key}
   & \pp{\tilde s}{\sigma,\topt}{f(\Pi_i) \le \val(\game_i)(s)+ \epsilon\ \vert\ \Pi_i\text{ is infinite and reaches $\tilde s$ infinitely often}}=1
%   \\
%    &  \label{eq:key2}
%  \pp{\tilde s}{\sigma,\topt}{f(\Pi_2) \le \val(\game_2)(s)+ \epsilon\ \vert\ \Pi_2\text{ is infinite and reaches $\tilde s$ infinitely often}}=1
  \enspace.
  \end{align}

  By symmetry, it is enough to show~\eqref{eq:key} when $i=1$.
  For that, we define a strategy $\sigma_1$ in $\game_1$ such that for every measurable event $\mathcal E_1$ in the game $\game_1$,
\begin{align}
  \label{eq:events}
  \pp{\tilde s}{\sigma_1,\topt_1}{\mathcal E_1}\ge \pp{\tilde s}{\sigma,\topt}{\text{$\Pi_1$ is
  infinite and }\Pi_1\in \mathcal E_1 }.
\end{align}
Denote by $\pref$ (respectively $\spref$) the prefix relation (respectively strict prefix) over finite or infinite plays.
The strategy $\sigma_1$ in $\game_1$ is defined as:
\begin{align*}
  \sigma_1(h_1)(a) = \pp{\tilde s}{\sigma,\topt}{h_1a\pref\Pi_1\ \vert\ h_1 \spref\Pi_1},
\end{align*}
if $\pp{\tilde s}{\sigma,\topt}{h_1 \spref \Pi_1}>0$ and otherwise $\sigma_1(h_1)$ is chosen arbitrarily. 
The event $h_1 \spref\Pi_1$ means that not only $h_1$ appears as a prefix
of the projection of the play on $\game_1$,
but moreover at least one more action has been played in $\game_1$ after that,
so $\sigma_1$ is equivalently defined as 
\[  \sigma_1(h_1)(a) = \pp{\tilde s}{\sigma,\topt}{h_1a\pref\Pi_1\ \vert\ \exists b \in \actions, h_1 b \pref\Pi_1}\enspace.
\]
Remark that in general, $\sigma_1$ is a mixed strategy.

We proceed with the proof of~\eqref{eq:events}.
Let $\mathfrak E$ be the set of measurable events $\mathcal E_1$\
in $\game_1$ for which \eqref{eq:events} holds. 
We prove first that $\mathfrak E$ contains all cylinders
$h_1(\states\actions)^\omega$ of $\game_1$,
which relies on the following inequalities:
\begin{align}
  \label{eq:sigma1}\pp{\tilde s}{\sigma_1,\topt_1}{h_1}& \ge \pp{\tilde s}{\sigma,\topt}{h_1\pref\Pi_1}\\
  \notag &\geq \pp{\tilde s}{\sigma,\topt}{\Pi_1\text{ is infinite and }\Pi_1 \in h_1(\states\actions)^\omega}.
\end{align}
We abuse the notation and denote $h_1$ the event $\{h_1 \text{ is a prefix of the play}\}$.
The second inequality is by definition of prefixes.
The inequality~\eqref{eq:sigma1} is proved by induction on the length of $h_1$.
When $h_1$ is the single initial state $\tilde s$ then both terms in~\eqref{eq:sigma1} are equal to $1$, and the inequality is an equality. 
Let $h_1 a r$ be a finite play in $\game_1$ and assume that~\eqref{eq:sigma1} holds for $h_1$.
There are two cases, depending who controls the last state of $h_1$, denoted $t$.
In case $t$ is controlled by Player $1$ then
\begin{align}
\notag\pp{\tilde s}{\sigma_1,\topt_1}{h_1ar} &=
\pp{\tilde s}{\sigma_1,\topt_1}{h_1} \cdot \sigma_1(h_1)(a)\cdot \psmall{t,a}{r} \\
\notag&\geq \pp{\tilde s}{\sigma,\topt}{h_1\pref\Pi_1}\cdot \sigma_1(h_1)(a) \cdot\psmall{t,a}{r} \\
\notag&= \pp{\tilde s}{\sigma,\topt}{h_1\pref\Pi_1}\cdot\pp{\tilde s}{\sigma,\topt}{h_1a\pref\Pi_1\ \vert\ h_1 \spref\Pi_1}\cdot\psmall{t,a}{r}\\
\label{eq:sneq}&\geq \pp{\tilde s}{\sigma,\topt}{h_1\spref\Pi_1}\cdot\pp{\tilde s}{\sigma,\topt}{h_1a\pref\Pi_1\ \vert\ h_1 \spref\Pi_1}\cdot\psmall{t,a}{r}\\
\notag&= \pp{\tilde s}{\sigma,\topt}{h_1a\pref\Pi_1} \cdot\psmall{t,a}{r}\\
\notag&= \pp{\tilde s}{\sigma,\topt}{h_1ar\pref\Pi_1}\enspace,
\end{align}
where the first and last equalities hold by definition of the probability measure,
the first inequality by induction hypothesis
and the second equality is by definition of $\sigma_1$.
The second inequality~\eqref{eq:sneq} holds because the event $h_1\spref\Pi_1$ is contained in the event 
$h_1\pref\Pi_1$. This inclusion and the corresponding inequality
might be strict: for example if $\Stay_{\geq 0}(\game_2)$ holds, i.e. if the play always stay in $\game_2$,
then the event $\tilde s\spref\Pi_1$ has probability $0$ while the event
$\tilde s\pref\Pi_1$ has probability $1$.

Now we prove inequality~\eqref{eq:sigma1},
in case $t$ is controlled by Player $2$.
For every finite play $h'_1$ in $\game_1$,
denote 
$C(h'_1)$ the set of finite plays $h'$ in $\game$ starting in $\tilde s$ and such that  
$\pi_1(h')=h'_1$ and $\last(h')\in \actions_1$.
Equivalently, $h'$ belongs $C(h'_1)$ if and only $\pi_1$ projects $h'$ on $h'_1$, but no strict prefix of $h'$
is projected on $h'_1$.

\begin{align*}
\pp{\tilde s}{\sigma_1,\topt_1}{h_1ar} &=
\pp{\tilde s}{\sigma_1,\topt_1}{h_1} \cdot \topt_1(h_1)(a)\cdot \psmall{t,a}{r} \\
&\geq \pp{\tilde s}{\sigma,\topt}{h_1\pref\Pi_1}\cdot \topt_1(h_1)(a) \cdot\psmall{t,a}{r} \\
&= 
\sum_{h' \in C(h_1)}
\pp{\tilde s}{\sigma,\topt}{h'}\cdot \topt_1(h_1)(a) \cdot\psmall{t,a}{r}\\
&= 
\sum_{h' \in C(h_1)}
\pp{\tilde s}{\sigma,\topt}{h'}\cdot \topt(h')(a) \cdot\psmall{t,a}{r}\\
&= 
\sum_{h' \in C(h_1)}
\pp{\tilde s}{\sigma,\topt}{h'ar}\\
&= 
\sum_{h'' \in C(h_1ar)}
\pp{\tilde s}{\sigma,\topt}{h''}\\
&= \pp{\tilde s}{\sigma,\topt}{h_1ar\pref\Pi_1}\enspace.
\end{align*}
The first equality is by definition of the probability measure.
The inequality is by induction hypothesis.
The second equality holds because the event
$h_1\pref\Pi_1$ is the disjoint union of the events $(h')_{h'\in C(h_1)}$:
if the projection of an infinite play $h$ to $\game_1$ starts with $h_1$,
then there is a single prefix of this play in $C(h_1)$,
this is the shortest (finite) prefix of $h$ whose projection in $\game_1$ is $h_1$.
The last equality holds by a similar argument.
The third equality is by definition of $\topt$.
The fourth equality is by definition of the probability measure.
To show the fifth equality, we establish $C(h_1ar)=\{h'ar\mid h'\in C(h_1)\}$.
We start with the inclusion $\{h'ar\mid h'\in C(h_1)\} \subseteq C(h_1ar)$.
Let $h'\in C(h_1)$. Since $\pi_1(h')=h_1$ and $\last(h')\in \actions_1$
then $h'$ and $h_1$ have the same last state, i.e. $t$.
And $t$ is controlled by Player $2$,
hence $t \neq \tilde s$. Thus $\last(h'ar)=\last(h') \in \actions_1$
and $\pi_1(h'ar)=\pi_1(h')ar$.
For the converse inclusion $C(h_1ar) \subseteq \{h'ar\mid h'\in C(h_1)\}$
take $h''\in C(h_1ar)$ and write $h''=h'a'r'$, where $a'$ and $r'$ are the last action and state of $h''$.
Since $\last(h'')\in\actions_1$
then $\pi_1(h'')=\pi_1(h')a'r' $.
Since $h''\in C(h_1ar)$ then $\pi_1(h'')=h_1ar$,
hence $\pi_1(h')=h_1$ and $a=a'$ and $r=r'$
hence $h''\in C(h_1)ar$. 
This completes the proof of the inequality~\eqref{eq:sigma1}.

%, using the definition of $\sigma_1$ and where  
Observe that $\mathfrak E$ is stable by finite disjoint unions, hence $\mathfrak E$ contains all finite disjoint unions of cylinders, which forms a boolean algebra.  Moreover $\mathfrak E$ is a monotone class, so we can apply the monotone class theorem (see for example \cite[Theorem 3.4]{billingsleyprob}).
This implies that $\mathfrak E$ contains the sigma-field that is generated by cylinders, which by definition is the set of all measurable events in the game $\game_1$. This completes the proof of~\eqref{eq:events}.

\medskip

Next we prove that
\begin{align}
  \label{eq:levi}
  \pp{\tilde s}{\sigma_1,\topt_1}{f \le \liminf_n \val(\game_1)(S_n)+ \epsilon}=1.  
\end{align}
Observe that due to the fact that $\topt_1$ is $\epsilon$-subgame-perfect and that $f$ is shift-invariant,
then
% the random variables in the sequence are upper bounded, that is 
for all $n\in\nat$,
\begin{align*}
  \ee {\tilde s}{\sigma_1,\topt_1}{f\ \vert\ S_0,A_0,\ldots,S_n}=\ee{S_n}{\sigma_1[S_0\cdots S_n],\topt_1[S_0\cdots S_n]} f\le \val(\game_1)(S_n)+\epsilon\enspace,
\end{align*}
and as a consequence,
\begin{equation}
\label{eq:liminf}
\liminf_n \ee {\tilde s}{\sigma_1,\topt_1}{f\ \vert\ S_0,A_0,\ldots,S_n} \le \liminf_n \val(\game_1)(S_n)+\epsilon\enspace.
\end{equation}
According to L\'evy's 0-1 law (see \eg \cite[Theorem 14.4]{wiliamspr}),
the sequence of random variables:
$(\ee {\tilde s} {\sigma_1,\topt_1} {f\ \vert\ S_0,A_0,\ldots,S_n})_{n\in\nat}$
 converges point-wise to the random~variable~$f(S_0A_0S_1\cdots)$.
As a consequence the left handside of~\eqref{eq:liminf} is almost-surely equal to $f$
and we get~\eqref{eq:levi}. 

Denote $\mathcal E_1$ the event
\[
\mathcal E_1
=\{
f > \val(\game_1)(\tilde s)+ \epsilon \text{ and $\tilde s$ is reached infinitely often}
\}\enspace.
\]
According to~\eqref{eq:levi},
$\pp{\tilde s}{\sigma_1,\topt_1}{\mathcal E_1}=0$.
We apply~\eqref{eq:events} to $\mathcal E_1$ and get
\[
\pp{\tilde s}{\sigma,\topt}{\text{$\Pi_1$ is
  infinite and }\Pi_1\in \mathcal E_1 }=0\enspace.
  \]
  By definition of $\mathcal E_1$,
  this last equality is equivalent to~\eqref{eq:key} with $i=1$.
\end{proof}

\subsection{Proof of  Theorem~\ref{thm:fundamentalineq}}

\begin{proof}[Proof of Theorem~\ref{thm:fundamentalineq}]
To prove the first statement~\eqref{eq:max} in Theorem~\ref{thm:fundamentalineq},
we combine the two lemmas
proved in the two previous sections in order to show:
\begin{align}
\label{eq:sall}
\forall \sigma, \ee{\tilde s}{\sigma,\topt}{f} \leq
  %\val(\game)(\tilde s) = 
  \max\set{\val(\game_1)(\tilde s),\val(\game_2)(\tilde s)}+ \epsilon\enspace.
\end{align}
The bound~\eqref{eq:sall} can be obtained as follows.
According to Lemma~\ref{rem:proj},
the three events $\stayone,\staytwo$ and $\switch$ partition
the set of infinite plays.
In case $\stayone$ occurs, Lemma~\ref{lem:fin} guarantees that the expected payoff is no more than 
$\val(\game_1)(\tilde s)+\epsilon$.
By symmetry, in case $\staytwo$ occurs,
the expected payoff is no more than 
$\val(\game_2)(\tilde s)+\epsilon$.
And in case $\switch$ occurs,
\Cref{lem:sinfty} guarantees that the payoff is almost-surely
no more than $\max\set{\val(\game_1)(\tilde s), \val(\game_2)(\tilde s)} + \epsilon$.
Thus~\eqref{eq:sall} holds.
The inequality 
\[
\val(\game)(\tilde s) \geq \max\set{\val(\game_1)(\tilde s),\val(\game_2)(\tilde s)}
\]
is clear, since Player~1 has more choice in $\game$ than he has in $\game_1$ and $\game_2$.
And $\epsilon$ can be chosen
arbitrarily small in~\eqref{eq:sall}, hence the first statement~\eqref{eq:max}
of \Cref{thm:fundamentalineq}.

\medskip

We proceed with the second statement of 
Theorem~\ref{thm:fundamentalineq}.
Assume that
\begin{align}
  \label{eq:max 1}
\val(\game_1)(\tilde s)\geq \val(\game_2)(\tilde s)\enspace. 
\end{align}
We have to show~\eqref{eq:last goal}, i.e.
\begin{align*}
\forall s \in \states,
  \val(\game)(\tilde s)=\val(\game_1)(\tilde s)\enspace.
\end{align*}
According to~\eqref{eq:max}, we already now that this equality holds for $\tilde s$,
and we shall extend it to all states $s\in\states$.

Recall that the merge strategy was defined only for plays that start in state $\tilde s$; we enlarge this definition, profiting from the assumption \eqref{eq:max 1}.
First, extend the definition of $\last(h)$ to any play $h$ that has visited $\tilde s$ at least once,
in which case $\last(h)$ denotes the action that is played right after the last visit of $h$ to $\tilde s$.
 Second, for all finite plays $h$ that end in a state controlled by Player $2$, 
\begin{align*}
\topt(h) \defeq \begin{cases}
  \topt_1(\pi_1(h)) & \text{ if $h$ never visited $\tilde s$ or  $\last(h)\in \actions_1$ }\\
  \topt_{2}(\pi_{2}(h)) & \text{ if $h$ has visited $\tilde s$ at least once and $\last(h)\in \actions_{2}$}.
\end{cases}
\end{align*}
The merge strategy is well-defined because if $h$ never visited $\tilde s$ or if $\last(h)\in\actions_1$ then both $h$ and $\pi_1(h)$ end with the same state, controlled by Player $2$. And if $h$ has visited $\tilde s$ at least once and  $\last(h)\in\actions_2$ then both $h$ and $\pi_2(h)$ end with the same state, controlled by Player $2$.

We prove that $\topt$ guarantees a payoff smaller than $\val(\game_1)(s)+\epsilon$ for every state $s$. Fix $\sigma$ a strategy for Player~1 in $\game$, and define $\sigma'$ to be the strategy that plays like $\sigma$ as long as the play does not reach the pivot state $\tilde s$.
Whenever the pivot state is reached, the strategy $\sigma'$ switches definitively to a strategy $\sopt_1$
that is optimal in the game $\game_1$, whose existence is guaranteed by the induction hypothesis. The plays consistent with $\sigma'$ and $\topt$ stay in the subgame $\game_1$. Since $\topt$ coincides with $\topt_1$ on plays staying in $\game_1$, and since $\topt_1$ is $\epsilon$-optimal in $\game_1$, we can write for all $s\in S$: 
\begin{align}
  \label{rq:rrr}    
  \ee s {\sigma',\topt} f=\ee s {\sigma',\topt_1} f\le\val(\game_1)(s) +\epsilon\enspace.
\end{align}
Let $h$ be a finite play that is consistent with $\sigma$ and $\topt$,
whose last state is $\tilde s$ and which does not visit $\tilde s$ before the last step.
Then
\begin{align*}
  \ee s {\sigma,\topt} {f\ \vert \ h} &= \ee {\tilde s} {\sigma[h],\topt[h]} f\\
  & \le \val(\game_1)(\tilde s)+\epsilon\\
  & \le \ee {\tilde s}{\sopt_1,\topt_1[h]} f +\epsilon\\
  &=\ee {\tilde s}{\sigma'[h],\topt_1[h]} f +\epsilon\\
  & = \ee s {\sigma',\topt} {f\ \vert\ h}+\epsilon\enspace.
\end{align*}
The first and third equalities hold because $f$ is prefix-independent.
The first inequality holds because the strategy $\topt[h]$ is $\epsilon$-optimal from state $\tilde s$,
for the following reason.
The strategy $\topt[h]$ coincides with the strategy obtained by merging $\topt_1[h]$ and $\topt_2$ on the pivot state $\tilde s$, both of which are $\epsilon$-subgame-perfect in the respective subgames. Since~\eqref{eq:sall} was proved for any merge of two $\epsilon$-subgame-perfect strategies, we can apply~\eqref{eq:sall} to the strategy $\topt[h]$, and conclude that the latter is $\epsilon$-optimal from state $\tilde s$.
The second inequality holds because $\sopt_1$ is optimal in $\game_1$.
The second equality holds because $\sigma'[h]=\sopt_1$.
Finally
\begin{align}\label{eq:gogo}
  \ee s {\sigma,\topt} {f\ \vert \ h} \leq \ee s {\sigma',\topt} {f\ \vert\ h}+\epsilon\enspace.
\end{align}

Since the strategies $\sigma$ and $\sigma'$ coincide on those plays that never reach $\tilde s$,
and~\eqref{eq:gogo} holds for all finite plays reaching $\tilde s$ for the first time,
then for every $s\in S$, 
\begin{align*}
  \ee s {\sigma,\topt} f \le \ee s {\sigma',\topt} f +\epsilon. 
\end{align*}
By using \eqref{rq:rrr} now we have that for all $s$,
\begin{align}
  \label{eq:merge is 2e optimal}
  \ee s {\sigma,\topt} f \le \val(\game_1)(s)+2\epsilon. 
\end{align}
This holds for  every strategy $\sigma$ and $\epsilon>0$ arbitrarily small, 
thus $\val(\game)(s)\leq \val(\game_1)(s)$. The converse inequality is obvious,
because Player $1$ has more freedom in $\game$ than in $\game_1$,
hence the second statement~\eqref{eq:last goal} 
of \Cref{thm:fundamentalineq}.
\end{proof}

\paragraph{Remarks about the merge strategy.}
We observe a byproduct of the proof of~\Cref{thm:fundamentalineq}, namely that \eqref{eq:merge is 2e optimal} has yielded $2\epsilon$-optimality of the merge strategy:
\begin{observation}
  \label{ob:merge is 2e optimal}
  The merge strategy $\topt$ constructed with $\epsilon$-subgame-perfect pieces is $2\epsilon$-optimal in the game $\game$.
\end{observation}

After this observation, since the merge strategy is obtained by merging two $\epsilon$-subgame-perfect strategies, a natural question to ask is whether $\topt$ is $2\epsilon$-subgame-perfect in the $\game$?
The answer is negative; consider the following simple example:

\includegraphics[width=.9\textwidth]{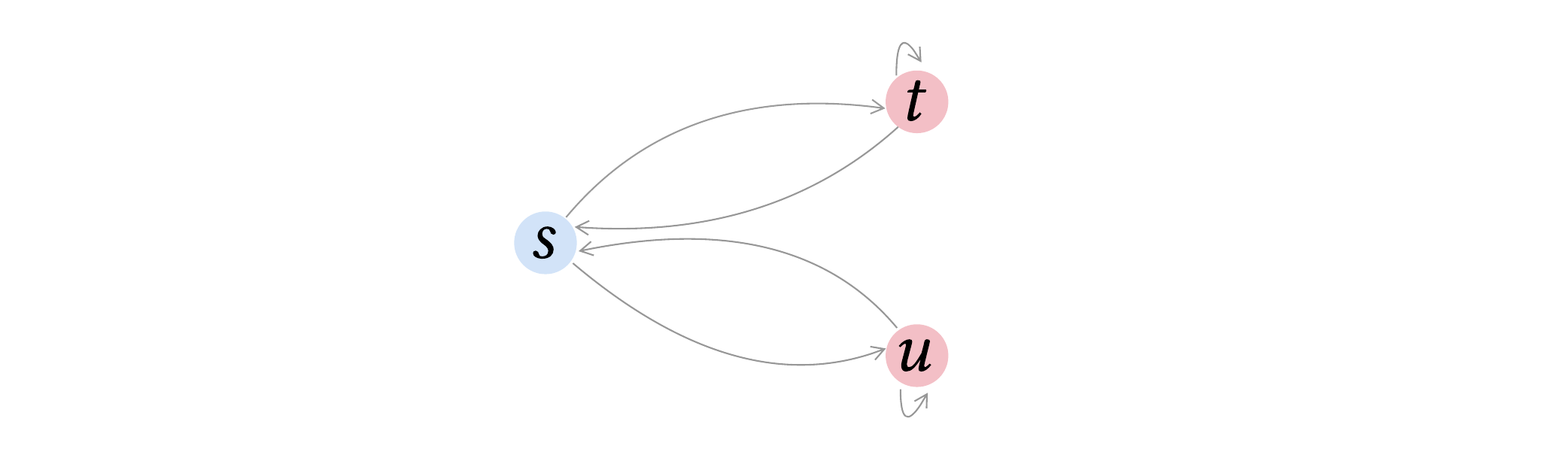}

The goal of Player~1 is to visit the state $t$ infinitely often (say that if he achieves this goal he receives a payoff 1, otherwise 0), and every action is deterministic. The blue states are controlled by Player~1, and the red ones by his opponent. In the subgame $\game_1$ we remove the action $s\to t$. In particular in the game $\game_1$ the positional strategy $\topt_1$ which chooses $u\to s$ and $t\to s$ is subgame-perfect. We can therefore use it to construct a merge strategy $\topt$. However this merge strategy is not $2\epsilon$-subgame-perfect, since in case Player~1 uses the suboptimal action $s\to u$, his opponent does not profit by taking the self-loop forever. 

%%% Local Variables:
%%% mode: latex
%%% TeX-master: "main"
%%% End:

\section{From One-player Games to Two-player Games}
\label{sec:onetotwo}

The construction of the merge strategy in the previous section reveals that games that are equipped with shift-invariant and submixing payoffs have the following interesting property. While they yield very simple optimal strategies for Player~1, they allow his opponent to recombine strategies that work for one-player games (also known as Markov decision processes) and use them in a two-player game. 

A general result allows to lift the existence of $\epsilon$-optimal strategies
from $\mathcal S$ in one-player games to two-player games.

An arena is said to be \emph{fully controlled by the minimizer} if all states are controlled by Player~2.
Fix a payoff function $f$ that is both shift-invariant and submixing.

%
%The proof of \Cref{thm:memory} can be carried on provided 
%that analogous to \Cref{sec:memory} and \Cref{lem:combifinite} do hold.
%.
%
%follows easily from \Cref{prop:reset is finite}
%and the following observation.
%\begin{lemma}\label{lem:combifinite}
%reveals a more general phenomenon.
%\begin{theorem}

\begin{definition}
Let $\mathcal{S}$ be a class of strategies for minimizer.

Say that the class $\mathcal{S}$ is \emph{stable by the reset operation}
if for every game $\game$ equipped with~$f$
and every strategy $\tau$ of minimizer in $\game$,
if $\tau$ belongs to $\mathcal{S}$ then
the reset strategy $\reset \tau$ belongs to $\mathcal{S}$ as well.

Say that the class $\mathcal{S}$ is \emph{stable by the merge operation}
if for every game $\game$ equipped with~$f$,
for every split $(\game_1,\game_2)$ of $\game$
and for every strategies $\topt_1$ and $\topt_2$ in $\game_1$ and $\game_2$,
if both $\topt_1$ and $\topt_2$ belong
to $\mathcal{S}$
then their merge $\topt$ belongs to $\mathcal{S}$ as well.
\end{definition}

Like in \Cref{prop:finite mem preservation},
say that the arena $\arena'$ is a \emph{restriction} of the arena $\arena$ if one gets $\arena'$ from $\arena$ by erasing some actions from some states.

\begin{theorem}\label{thm:generaltransfer}
Let $f$  a shift-invariant payoff function,
$\mathbb A$  a family of arenas that are closed under restrictions and 
$\mathbb{S}$  a family of strategies for minimizer
which are stable by both reset and merge operations. 

Assume that in every game $(A,f)$ with $A\in\mathbb A$
that is fully controlled by the minimizer,
for every $\epsilon>0$ the minimizer has an $\epsilon$-optimal strategy that belongs to $\mathcal{S}$.
Then, in every two-player game $(A,f)$ with $A\in\mathbb A$,
 the minimizer has an $\epsilon$-subgame perfect strategy that belongs to $\mathbb{S}$.

The statement holds for $\epsilon=0$ as well, that is:
assume that in all games $(A,f)$ with $A\in\mathbb A$ 
that is fully controlled by the minimizer,
the minimizer has an optimal strategy that belongs to $\mathbb{S}$.
Then, in every two-player game $(A,f)$ with $A\in\mathbb A$,
the minimizer has a subgame perfect strategy that belongs to $\mathbb{S}$.
\end{theorem}
\begin{proof}
Let $\game=(A,f)$ with $A\in\mathbb A$.
  The proof of both statements is by induction on $N(\game)$, 
  %by defining the smaller games $\game_1$ and $\game_2$ 
  as in the proof of the main theorem in the previous section. 

  The base of the induction follows from the assumption about games fully controlled by the minimizer, since we can give to the minimizer
  the control of states in which the maximizer has a single action,
  without changing the value of the game.
  
  When $N(\game)>0$, the induction step is performed using a split $(\game_1,\game_2)$ of $\game$ on a pivot state $\tilde s$.
  Remark that both arenas belong to $\mathbb A$,
  therefore the induction hypothesis for the first (resp. the second) statement
  says that for every $\epsilon > 0$,
  there are two strategies
$\topt_1$ and $\topt_2$ in the games $\game_1$ and $\game_2$, respectively,
which belong to $\mathcal S$ and are
$\epsilon$-subgame perfect
(resp. subgame perfect)
  in their respective subgames.
  Since $\mathcal{S}$ is stable by the merge operation,
  then the strategy $\topt$ obtained by merging $\topt_1$ and $\topt_2$
  also belongs to $\mathbb S$.
  
We carry over the induction step for the first (resp. the second) statement.
According to \Cref{ob:merge is 2e optimal}, 
$\topt$ is $2\epsilon$-optimal (resp. is optimal).
We apply \Cref{lem:subgame} to $\topt$ which guarantees that the reset strategy obtained from
$\topt$ is 
$4\epsilon$-optimal (resp. is optimal). 
Moreover by hypothesis this strategy belongs to $\mathcal S$.
\end{proof}

\section{The Finite Memory Transfer Theorem}
\label{sec:memory}

We give the proof of \Cref{thm:memory} that was announced in the introduction.

\memory*

\Cref{thm:memory} follows from \Cref{thm:generaltransfer}
and the following results,
which establish that the class of finite-memory strategies
is stable by the reset (\Cref {prop:finite mem preservation}) 
and merge (\Cref{lem:combifinite}) operations.

\begin{proposition}
  \label{prop:finite mem preservation}
  Let $\mathbb A$ be a family of arenas that are closed under restrictions and $f$ a shift-invariant payoff function. If for games whose arena is in $\mathbb A$ and whose payoff function is $f$, and for every $\epsilon>0$, Player~1 (respectively Player~2) has an $\epsilon$-optimal strategies $\sigma$ with finite memory, then he also has an $\epsilon$-subgame-perfect strategies with finite memory, namely the reset strategies $\reset\sigma$.
  This holds as well for optimal strategies, i.e. if $\epsilon = 0$.
\end{proposition}
\begin{proof}
  Let $\arena\in\mathbb A$ be an arena. Remove the actions of Player~1 that are not locally optimal (with respect to the payoff function $f$) to get a restriction $\arena'$. From the hypothesis, it follows that there are $\epsilon$-optimal strategies in $(\arena',f)$ that have finite memory, and consequently there are $\epsilon$-optimal strategies in $(\arena,f)$ that have finite memory and are locally optimal. According to \Cref{lem:subgame},
  the strategy $\reset\sigma$ is $2\epsilon$-subgame-perfect, and \Cref{prop:reset is finite} implies that it has finite memory. 
\end{proof}

\begin{lemma}\label{lem:combifinite}
Let $(\game_1,\game_2)$ be a split of a game $\game$ on a pivot state $\tilde s$.
Let  $\topt_1$ and $\topt_2$ two strategies for Player~2 in 
$\game_1$ and $\game_2$, respectively.
If both $\topt_1$ and $\topt_2$ have finite-memory then $\topt$ has finite memory as well.
\end{lemma}
\begin{proof}
The strategies $\topt_1$ and $\topt_1$ with finite memory are given by the transducers:
\begin{align*}
 (\mem_1,\init_1,\up_1,\out_1)\qquad \text{and}\qquad(\mem_2,\init_2,\up_2,\out_2),
\end{align*}
for Player 2 in $\game_1$ and $\game_2$ respectively. 

The strategy $\topt$ obtained by merging $\topt_1$ and $\topt_2$ is also a finite-memory strategy, whose memory is
\begin{align*}
 \mem\defeq\set{1,2}\times \mem_1\times \mem_2.
\end{align*}
The initial memory state in state $s$ is $(1,\init_1(s),\init_2(s))$. The updates on the components $\mem_1$ and $\mem_2$ are performed with $\up_1$ and $\up_2$ respectively.  The first component is updated only when the play leaves the pivot state $\tilde s$; it is switched to $1$ or $2$ depending whether Player $1$ chooses an action in $\actions_1$ or $\actions_2$.  The choice of action, or the output,  depends on the first component: in memory state $(b,m_1,m_2)$ the action played by $\topt$ is $\out_b(m_b)$.
\end{proof}

The finite-memory property is preserved when passing from $\sigma$ to $\hat\sigma$
that is if the strategy $\sigma$ has finite memory to begin with, so will the strategy $\reset\sigma$. First we define precisely what we mean by finite memory strategy.

A strategy $\sigma$ is said to have \emph{finite memory} if it is given using a transducer, namely it is a tuple:
\begin{align*}
  \underbrace{\mem}_{\text{a finite set}},\qquad \underbrace{\init\st \states\to\mem}_{\text{memory initialiser}},\qquad \underbrace{\up\st\mem\times\actions\times\states\to\mem}_{\text{update function}},\qquad \underbrace{\out\st\mem\to\distrib\actions}_{\text{output function}}.
\end{align*}
The map $\init$ and $\up$ are used to initialise the memory and update it, as the game unfolds: after the finite play $s_0a_0\cdots s_n$ has unfolded, the transducer reaches the memory state $m_n\in\mem$ which is defined inductively as:
\begin{align*}
  m_0&\defeq \init(s_0), \text{and}\\
  m_k&\defeq \up(m_k,a_{k+1},s_{k+1}). 
\end{align*}
The output function is used to choose the action that the strategy plays, \ie
\begin{align*}
  \sigma(s_0\cdots s_n)=\out(m_n).
\end{align*}

\begin{proposition}
  \label{prop:reset is finite}
  Let $\epsilon>0$.
  If $\sigma$ is a finite memory strategy and is $\epsilon$-optimal
  then the $\epsilon$-reset of $\sigma$ has finite-memory as well.
\end{proposition}
\begin{proof}
  Let $\sigma$ be a finite memory strategy, that is given by the tuple
  \begin{align*}
    (\mem,\init,\up,\out),
  \end{align*}
  and let $\epsilon$ be such that $\sigma$ is $\epsilon$-optimal, which fixes a reset strategy $\reset\sigma$.

  Without loss of generality we can assume that the strategy is such that its memory state identifies the current state in the game, in other words assume that $\mem$ can be partitioned into:
  \begin{align*}
    \mem = \biguplus_{s\in\states}\mem_s,
  \end{align*}
  such that for any finite play $s_0\cdots s_n$, if $m_1,\ldots,m_n$ is the sequence of memory states of the transducer of $\sigma$ during this play, then
  \begin{align*}
    m_n\in\mem_{s_n}.
  \end{align*}

  We gather the subset of memory states where drops occur as follows. For $s\in\states$ and $m\in\mem_s$, denote by $\sigma_m$ the strategy that is the same as $\sigma$ except that the initial memory state for $s$ is $m$ instead of $\init(s)$. Define the subset of memory states where drops occur $\mathcal D\subset\mem$ as
  \begin{align*}
    \mathcal D\defeq \set{m\in\mem_s\st s\in\states\text{ and $\sigma_m$ is not $2\epsilon$-optimal from state $s$}}.
  \end{align*}
  Construct the finite memory strategy $\sigma'$ that avoids the memory states in $\mathcal D$ as follows. For any $s\in\states$ and $m\in\mem_s\cap \mathcal D$, since $\sigma$ is $\epsilon$-optimal, $m\ne \init(s)$. In the strategy $\sigma'$ modify the function $\up$ in such a way that all the transitions that lead to $m$ are redirected to the state $\init(s)$ instead (the memory is reset). Do this simultaneously for any pair $(s,m)$ as above. Comparing the definition of $\reset\sigma$ and $\sigma'$ we conclude that they coincide. 
\end{proof}

\paragraph{On the size of the memory.} How large is the memory $\mem_\game$ needed by Player 2 to play optimally in some $\game=(\mathcal{A},\pay)$?  Every deterministic and stationary strategy $\sigma$ for Player 1 in $\game$ induces a game $\game_\sigma$ that is controlled by Player 2.  Let $\mathfrak M$ be the maximal memory size required by Player 2 to play optimally in the games $G_\sigma$.  According to the proof of the theorem above, the memory $\mem_\game$ needed by Player 2 to play optimally in $\game$ is of size $2\cdot |\mem_{\game_1}| \cdot |\mem_{\game_2}|$. By induction we derive the following bound:
\begin{align*}
  |\mem_\game| \le (2\mathfrak M)^{2^{\sum_s |\actions(s)|}}.
\end{align*}
When $\mathfrak M=1$, \ie when Player 2 has deterministic and stationary strategies in games he controls, then in~\cite{concur} it is shown that the same holds for two player games as well, hence the upper-bound can be downsized to 1. In the general case where $\mathfrak M\ge 2$, we do not have examples where the memory size required by Player 2 to play optimally has the same order of magnitude as the upper bound above.

%%% Local Variables:
%%% mode: latex
%%% TeX-master: "main"
%%% End:

\section*{Acknowledgments}
We are very grateful to Pierre Vandenhove for finding an error in the previous version of the proof of \Cref{lem:fin}. This work was supported by the ANR projet "Stoch-MC" and the LaBEX "CPU".
\bibliographystyle{alpha} 
\bibliography{2jpos} 
\end{document}